\documentclass[journal]{IEEEtran}

\usepackage{amsmath,graphicx}
\usepackage{amsfonts}
\usepackage{amssymb}
\usepackage{epsfig}
\usepackage{mathrsfs}
\usepackage{graphicx}
\usepackage{algorithm}
\usepackage{caption}
\usepackage{subcaption}
\usepackage{xcolor}
\usepackage{cite}
\usepackage{url}

\usepackage{balance}

\newtheorem{theorem}{Theorem}[section]
\newtheorem{lemma}[theorem]{Lemma}

\newenvironment{proof}[1][Proof]{\begin{trivlist}
\item[\hskip \labelsep {\bfseries #1}]}{\qed \end{trivlist}}
\newenvironment{definition}[1][Definition]{\begin{trivlist}
\item[\hskip \labelsep {\bfseries #1}]}{\end{trivlist}}

\newcommand\qed{$\blacksquare$}

\DeclareMathOperator{\rank}{\operatorname{rank}}
\DeclareMathOperator{\trace}{\mathrm{tr} \, }
\DeclareMathOperator{\diag}{\mathrm{diag} \, }

\DeclareMathOperator{\diff}{ \, \mathrm{d} \!}

\newsavebox\myboxA
\newsavebox\myboxB
\newlength\mylenA

\newcommand{\oline}[1]{\mkern 1.5mu\overline{\mkern-1.5mu#1\mkern-1.5mu}\mkern 1.5mu}

\def\mA{\mbox{$\mathbf{A}$}}
\def\mB{\mbox{$\mathbf{B}$}}
\def\mBc{\mbox{$\oline{\mathbf{B}}$}}

\def\mD{\mbox{$\mathbf{D}$}}
\def\mDc{\mbox{$\oline{\mathbf{D}}$}}
\def\mG{\mbox{$\mathbf{G}$}}

\def\mI{\mbox{$\mathbf{I}$}}
\def\mL{\mbox{$\mathbf{L}$}}
\def\mQ{\mbox{$\mathbf{Q}$}}

\def\mT{\mbox{$\mathbf{T}$}}
\def\mU{\mbox{$\mathbf{U}$}}

\def\mY{\mbox{$\mathbf{Y}$}}

\def\mSigma{\mbox{$\mathbf{\Sigma} \kern .08em$}}
\def\mLambda{\mbox{$\mathbf{\Lambda} \kern .08em$}}

\newcommand{\A}{{\cal A}}
\newcommand{\E}{{\cal E}}
\newcommand{\G}{{\cal G}}

\newcommand{\I}{{\cal I}}

\newcommand{\V}{{\cal V}}

\newcommand{\abs}[1]{\left\vert#1\right\vert}

\def\N{\text{\mbox{${\cal N}$}}}
\def\S{\text{\mbox{${\cal S}$}}}
\def\F{\text{\mbox{${\cal F}$}}}

\def\H{\text{\mbox{${\cal H}$}}}
\def\V{\text{\mbox{${\cal V}$}}}
\def\E{\text{\mbox{${\cal E}$}}}
\def\G{\text{\mbox{${\cal G}$}}}
\def\B{\text{\mbox{${\cal B}$}}}
\def\D{\text{\mbox{${\cal D}$}}}

\def\b0{\text{\mbox{\boldmath $0$}}}
\def\ba{\text{\mbox{\boldmath $a$}}}
\def\bb{\text{\mbox{\boldmath $b$}}}

\def\bff{\text{\mbox{\boldmath $f$}}}

\def\bg{\text{\mbox{\boldmath $g$}}}

\def\bn{\text{\mbox{\boldmath $n$}}}

\def\br{\text{\mbox{\boldmath $r$}}}
\def\bs{\text{\mbox{\boldmath $s$}}}
\def\bu{\text{\mbox{\boldmath $u$}}}

\def\bx{\text{\mbox{\boldmath $x$}}}
\def\bxs{\text{\mbox{$\boldsymbol{x}_{\cal S}$}}}
\def\by{\text{\mbox{\boldmath $y$}}}

\def\buno{\text{\mbox{\boldmath $1$}}}

\def\bdelta{\text{\mbox{\boldmath $\delta$}}}

\def\bpsi{\text{\mbox{\boldmath $\psi$}}}

\ifCLASSINFOpdf
\else
\fi
\let\svthefootnote\thefootnote

\begin{document}

\title{Signals on Graphs: \\ Uncertainty Principle and Sampling}

\author{Mikhail~Tsitsvero,
        Sergio~Barbarossa,~\IEEEmembership{Fellow,~IEEE,}
        and~Paolo~Di~Lorenzo,~\IEEEmembership{Member,~IEEE}
\thanks{M. Tsitsvero, S. Barbarossa are with the Department of
Information Eng., Electronics and Telecommunications, Sapienza University of Rome, Rome, Italy,
e-mail: tsitsvero@gmail.com, sergio.barbarossa@uniroma1.it

P. Di Lorenzo is with the Department of Engineering, University of Perugia, Perugia, Italy,
e-mail: paolo.dilorenzo@unipg.it}}

\maketitle

\begin{abstract}
In many applications, the observations can be represented as a signal defined over the vertices of a graph. The analysis of such signals requires the extension of standard signal processing tools. In this work, first, we provide a class of graph signals that are maximally concentrated on the graph domain and on its dual. Then, building on this framework, we derive an uncertainty principle for graph signals and illustrate the conditions for the recovery of band-limited signals from a subset of samples. We show an interesting link between uncertainty principle and sampling and propose alternative signal recovery algorithms, including a generalization to frame-based reconstruction methods. After showing that the performance of signal recovery algorithms is significantly affected by the location of samples, we suggest and compare a few alternative sampling strategies. Finally, we provide the conditions for perfect recovery of a useful signal corrupted by sparse noise, showing that this problem is also intrinsically related to vertex-frequency localization properties.
\end{abstract}


%
\IEEEpeerreviewmaketitle

\section{Introduction}
\label{sec:intro}
\IEEEPARstart{I}{n} many applications, from sensor to social networks, transportation systems, gene regulatory networks or big data, the signals of interest are defined over the vertices of a graph \cite{shuman2013emerging}, \cite{sandryhaila2014big}. 
Over the last few years, a series of papers produced a significant advancement in the development of tools for the analysis of signals defined over a graph, or graph signals for short \cite{shuman2013emerging}, \cite{pesenson2010sampling}, \cite{sandryhaila2013discrete}. One of the unique features in graph signal processing is that the analysis tools come to depend on the graph topology. This paves the way to a plethora of methods, each emphasizing different aspects of the problem. A central role is played by spectral analysis of graph signals, which passes through the introduction of the Graph Fourier Transform (GFT). 
Alternative definitions of GFT have been proposed, see, e.g., \cite{shuman2013emerging,sandryhaila2014big, pesenson2010sampling, pesenson2008sampling, sandryhaila2013discrete, zhu2012approximating, chen2015discrete}, looking at the problem from  different perspectives: \cite{shuman2013emerging}, \cite{pesenson2008sampling}, \cite{zhu2012approximating} apply to {\it undirected} graphs and build on the spectral clustering properties of the Laplacian eigenvectors and the minimization of the $\ell_2$-norm graph total variation; \cite{sandryhaila2013discrete}, \cite{chen2015discrete} define a GFT for {\it directed} graphs, building on  the interpretation of the adjacency operator as the graph shift operator, which lies at the heart of all linear shift-invariant filtering methods for graph signals \cite{Puschel1}, \cite{Puschel2}.\let\thefootnote\relax\footnote{
Part of this work was presented at the 23-rd European Signal Processing Conf. (EUSIPCO), Sep. 2015 \cite{TsitsveroEusipco15} and at the 49-th Asilomar Conf. on Signals, Systems and Computers, Nov. 2015 \cite{Tsitsvero2015}.}
\addtocounter{footnote}{-1}\let\thefootnote\svthefootnote Building on  \cite{chung2005laplacians}, one could also introduce a GFT for directed graphs based on the eigendecomposition of the {\it modified} Laplacian for directed graphs introduced in \cite{chung2005laplacians}.

After the introduction of the GFT,  an uncertainty principle for graph signals was derived in \cite{agaskar2013spectral} and, more recently, in \cite{pasdeloup2015toward}, \cite{benedettograph}, \cite{koprowski2015finite}. The aim of these works was to establish a link between the spread of a signal on the vertices of the graph and the spread of its spectrum, as defined by the GFT, on the dual domain. The further fundamental contribution was the formulation of a sampling theory aimed at finding the conditions for recovering a graph signal from a subset of samples: A seminal contribution was given in \cite{pesenson2008sampling}, later extended in \cite{narang2013signal}, \cite{anis2014towards} and, very recently, in  \cite{chen2015discrete}, \cite{wang2014local}, \cite{marquez2015}.
In the following, after introducing the notation, we briefly recall the background of graph signal processing and then we highlight the specific contributions of this paper.

\subsection{Notation and Background}
We consider a graph $\G = (\V, \E)$ consisting of a set of $N$ nodes $\V = \{1,2,..., N\}$, along with a set of weighted edges $\E=\{a_{ij}\}_{i, j \in \V}$, such that $a_{ij}>0$, if there is a link from node $j$ to node $i$, or $a_{ij}=0$, otherwise. The symbol $|\S|$ denotes the cardinality of set $\S$, i.e., the number of elements of $\S$.  The adjacency matrix $\mA$ of a graph is the collection of all the weights $a_{ij}, i, j = 1, \ldots, N$.  The degree of node $i$ is $k_i:=\sum_{j=1}^{N}a_{ij}$. The degree matrix is a diagonal matrix having the node degrees on its diagonal: $\mathbf{K} = \diag \{ k_1, k_2, ... , k_N \}$.  
The combinatorial Laplacian matrix is defined as $\mathbf{L} = \mathbf{K}-\mathbf{A}$.
In the literature, it is also common to use the normalized graph Laplacian matrix $\mathbf{\mathcal{L}}=\mathbf{K}^{-1/2}\mathbf{L}\mathbf{K}^{-1/2}$. A signal $\bx$ over a graph $\G$ is defined as a mapping from the vertex set to the set of complex numbers, i.e. $\bx: \V \rightarrow \mathbb{C}$. We denote by $\| \cdot \|$  the $\ell_2$-norm of a signal, i.e. $ \| \bx \|^2 = \sum_{i \in \V} \abs{x_i}^2 $. We recall now the basic background for better clarifying the contributions of our work. 

Let us introduce the eigen-decomposition of the Laplacian matrix
\begin{equation}
\label{eq:L=sum xi ui}
\mathbf{L} = \mathbf{U} \mathbf{\Xi} \mathbf{U}^* = \sum_{i=1}^N \xi_i \bu_i \bu^*_i,
\end{equation}
where $\mathbf{\Xi}$ is a diagonal matrix with non-negative real eigenvalues $\{ \xi_i \}$ on its diagonal and $\{\bu_i \}$, $i=1, \ldots, N$, are the real-valued orthonormal eigenvectors; the symbol $(\cdot)^*$ denotes conjugate transpose. The Graph Fourier Transform $\mathbf{\hat{\bx}}$ of a signal $\bx$ defined over an undirected graph has been defined in \cite{pesenson2008sampling},  \cite{pesenson2010sampling}, \cite{shuman2013emerging}, \cite{zhu2012approximating},
as 
\begin{equation}
\label{eq::gft}
\mathbf{\hat{\bx}} = \mathbf{U}^* \bx
\end{equation}
where $\mU$ is the unitary matrix whose columns are the Laplacian eigenvectors.
One of the motivations for projecting the signal $\bx$ onto the subspace spanned by the eigenvectors of $\mL$, as in (\ref{eq::gft}), is that 
these eigenvectors encode some of the graph topological features. For example, they are known for exhibiting spectral clustering properties  \cite{Chung1997}, \cite{von2007tutorial}. Hence, the GFT defined in (\ref{eq::gft}) is useful for emphasizing clustered signal components, i.e. signals that are smooth within a cluster, but are allowed to vary arbitrarily across different clusters. 
In this work, we assume the GFT to be defined as in (\ref{eq::gft}), where $\mathbf{U}$ is only required to be a unitary matrix. In most numerical examples we assume $\mathbf{U}$ to be composed by the eigenvectors of the Laplacian matrix, but all derivations are not restricted to that choice. This means that all theoretical findings are valid for any mapping from primal to dual domain described by a unitary operator $\mU$. 
Given a subset of vertices $\S \subseteq \V$, 
we define a vertex-limiting operator as a diagonal matrix  $\mathbf{D}_{\S}$ such that
\begin{equation}
\label{D}
\mathbf{D}_{\S} = {\rm Diag}\{\buno_{\S}\},
\end{equation}
where $\buno_{\S}$ is the set indicator vector, whose $i$-th entry is equal to one, if  $i \in \S$, and zero otherwise. Similarly, given the unitary matrix $\mU$ used in (\ref{eq::gft}), and a subset of indices $\F \subseteq \V^*$, where $\V^* = \{1, \ldots,N\}$ denotes the set of all frequency indices, we introduce the operator
\begin{equation}
\label{lowpass_operator}
\mB_{\F} = \mathbf{U \Sigma_{\F} U}^*,
\end{equation}
where $\mathbf{\Sigma_{\F}}$ is a diagonal matrix defined as $\mathbf{\Sigma_{\F}}= {\rm Diag}\{\buno_{\F}\}$. The role of $\mB_{\F}$ is to project a vector $\bx$ onto the subspace spanned by the columns of $\mU$ whose indices belong to $\F$. 
It is immediate to check that both matrices $\mathbf{D}_{\S}$ and $\mathbf{B}_{\F}$ are self-adjoint and idempotent, so that they represent orthogonal projectors. In the sequel, 
$\D_{\S}$ denotes the set of all $\S$-vertex-limited signals, i.e. satisfying $\mD_\S\, \bx = \bx$, whereas
 $\B_{\F}$ denotes the set of all $\F$-band-limited signals, i.e. satisfying $\mB_\F\, \bx = \bx$.
In the rest of the paper, whenever there will be no ambiguities in the specification of the sets, we will drop the subscripts referring to the sets, to avoid overcrowded symbols.
Given a set $\S$, we denote its complement set as $\oline{\S}$, such that $\V=\S \cup \oline{\S}$ and $\S \cap \oline{\S}=\emptyset$. Correspondingly, we define the vertex-projector onto $\oline{\S}$ as $\oline{\mathbf{D}}$. Similarly, the projector onto the complement set $\oline{\F}$ is denoted by $\oline{\mathbf{B}}$.
\\

\subsubsection{Uncertainty principle}
A fundamental property of continuous-time signals is the Heisenberg uncertainty principle, stating that there is a basic trade-off between the spread of a signal in time and and the spread of its spectrum in frequency. In particular, a continuous-time signal cannot be perfectly localized in both time and frequency domains (see, e.g., \cite{Folland1997} for a survey on the uncertainty principle). More specifically, given a continuous-time signal $x(t)$ and its Fourier transform $X(f)$, introducing the time spread 
\begin{equation}
\label{Delta_T}
\Delta_t^2=\frac{\int_{-\infty}^{\infty}(t-t_0)^2|x(t)|^2 \diff t}{\int_{-\infty}^{\infty}|x(t)|^2 \diff t}
\end{equation}
with
\begin{equation*}
t_0=\frac{\int_{-\infty}^{\infty} t\, |x(t)|^2 \diff t}{\int_{-\infty}^{\infty}|x(t)|^2 \diff t}
\end{equation*}
and the frequency spread
\begin{equation}
\label{Delta_F}
\Delta_f^2=\frac{\int_{-\infty}^{\infty}(f-f_0)^2|X(f)|^2 \diff f}{\int_{-\infty}^{\infty}|X(f)|^2 \diff f},
\end{equation}
with 
\begin{equation*}
f_0=\frac{\int_{-\infty}^{\infty} f\, |X(f)|^2 \diff f}{\int_{-\infty}^{\infty}|X(f)|^2 \diff f},
\end{equation*}
the uncertainty principle states that
\begin{equation*}
\Delta_t^2 \Delta_f^2 \ge \frac{1}{(4 \pi)^2}.
\end{equation*}
After the introduction of the GFT, an uncertainty principle for signals defined over undirected connected graphs was derived in {\cite{agaskar2013spectral}}. 
In particular, denoting by $d(u, v)$ the {\it geodesic} distance between nodes $u$ and $v$, i.e. the length of the shortest path connecting $u$ and $v$, the spread of a vector $\bx$ in the vertex domain was defined in \cite{agaskar2013spectral} as
\begin{equation}
\label{delta_t_AgaskarLu}
\Delta_g^2:= \min_{u_0 \in {\cal V}} \frac{1}{\|\bx\|^2}\,\bx^*\mathbf{P}^2_{u_0} \bx
\end{equation}
where $\mathbf{P}_{u_0}:={\rm diag}(d(u_0, v_1), d(u_0, v_2), \ldots, d(u_0, v_N))$.
Similarly, the spread in the GFT domain was defined as
\begin{equation}
\label{delta_f_AgaskarLu}
\Delta_s^2:=  \frac{1}{\|\bx\|^2}\, \sum_i \xi_i \abs{\hat{x}_i}^2,
\end{equation}
where $\xi_i$ was defined in (\ref{eq:L=sum xi ui}).
The two definitions of spread in the graph and its dual domain given in (\ref{delta_t_AgaskarLu}) and  (\ref{delta_f_AgaskarLu}) are the graph counterparts of formulas (\ref{Delta_T}) and (\ref{Delta_F}) for continuous-time signals. 
In \cite{agaskar2013spectral}, it was studied the tradeoff between the signal spread on the graph and on its
spectral (dual) domain, i.e. between (\ref{delta_t_AgaskarLu}), for a given value of $u_0$, i.e. without performing the minimization operation,  and  (\ref{delta_f_AgaskarLu}).
\subsubsection{Sampling}
One of the basic issues in graph signal processing is sampling, whose goal is to find the conditions for recovering a band-limited (or approximately band-limited) graph signal from a subset of values and to devise suitable sampling and recovery strategies. 
More specifically, a band-limited graph signal can be represented as 
\begin{equation}
\label{x=Us}
\bx=\mU \bs,
\end{equation}
where $\mU$ is an appropriate basis and $\bs$ is sparse. Typically, $\mU$ coincides with the matrix whose columns are the eigenvectors of $\mL$. If we denote by $\S \subseteq \V$ the sampling subset, the sampled signal can be represented as
\begin{equation}
\label{eq::sampling equation}
\bx_{\S}=\mathbf{D}_{\S}\, \bx=\mathbf{D}_{\S} \,\mU \bs,
\end{equation}
where $\mathbf{D}_{\S}$ is defined as in (\ref{D}). The problem of recovering a band-limited signal from its samples is then equivalent to the problem of solving system (\ref{eq::sampling equation}), by exploiting the sparsity of $\bs$. This problem was addressed, for example, in \cite{pesenson2008sampling}, \cite{narang2013signal},  \cite{wang2014local}, and \cite{chen2015discrete}. 
Alternative recovery strategies have been proposed, either iterative \cite{narang2013localized}, \cite{wang2014local}, or not \cite{chen2015discrete}. In \cite{pesenson2008sampling}, \cite{wang2014local}, frame-based recovery algorithms have also been proposed.

A key important remark is that the sampling strategy, i.e. the identification of the sampling set $\S$, plays a key role in the performance of the recovery algorithms, as it affects the conditioning of system (\ref{eq::sampling equation}). It is then particularly important to devise strategies to optimize the selection of the sampling set. This problem is conceptually similar to the problem known in the literature as {\it experimental design}, see, e.g., \cite{Steinberg-experimental_design1984, avron2013faster, Ranieri-Cheibra-Vetterli-2014}. Sampling strategies for graph signals were  proposed in \cite{anis2014towards} and, more recently, in \cite{chen2015discrete}.

\subsection{Contributions}
The main contribution of this paper is to present a holistic framework that unifies uncertainty principle and sampling, building on the identification of the class of graph signals that are maximally concentrated over the graph and dual domain. The specific contributions are listed below.
\subsubsection{Uncertainty principle}
The definitions of spread in the graph and dual domain given in (\ref{delta_t_AgaskarLu}) and (\ref{delta_f_AgaskarLu}), as suggested in \cite{agaskar2013spectral}, are reminiscent of the formulas (\ref{Delta_T}) and (\ref{Delta_F}) valid for continuous-time signals. They are both based on second order moments of the signal distribution over the graph domain and on its dual. However, when dealing with graph signals, there is an important distinction to be pointed out with respect to time signals: While time (or frequency) is a metric space, with a well defined notion of distance, the graph domain is not a metric space. The vertices of a graph may represent, for example, molecules and the signal may be the concentration of a molecule in a given mixture. In cases like this, it is not obvious how to define a distance between vertices. Since in (\ref{delta_t_AgaskarLu}) the definition of distance enters directly in the computation of the spread,  it turns out that the uncertainty principle comes to depend on the specific definition of distance over the graph. The definition of distance given in \cite{agaskar2013spectral} makes perfect sense, but as pointed out by the authors themselves, it is not the only possible choice.   When dealing with graphs, other definitions of distance have been proposed in the literature, including the resistance distance \cite{KleinRandic93} and the diffusion distance \cite{CoifmanMaggioni06}. An open question arises, for example, in the presence of multiple shortest paths having the same distance between two vertices. In such a case, using the definition (\ref{delta_t_AgaskarLu}), the presence of multiple paths having the same distance does not affect the computation of the spread. However, the presence of multiple paths might indicate an easier way for the information to flow through the network. In fact, using the definition of resistance distance suggested in \cite{KleinRandic93}, the distance between two nodes comes to depend on the number of shortest paths with the same distance connecting them. To avoid all the shortcomings associated with the definition of distance over a graph, in this paper we use an alternative definition of spread and derive an uncertainty principle that does not require any additional definition of distance. More specifically, we take inspiration from the seminal works of Slepian, Landau and Pollack {\cite{Slepian:1961:PSW}}, {\cite{landau1961prolate}}, on prolate spheroidal wave-functions. In those works, the effective duration $T$ of a continuous-time signal centered around a time instant $t_0$ was defined as the value  such that the percentage of energy falling in the interval $[t_0-T/2, t_0+T/2]$ assumes a specified value $\alpha^2$, i.e.
\begin{equation*}
\frac{\int_{t_0-T/2}^{t_0+T/2}  |x(t)|^2 \diff t}{\int_{-\infty}^{\infty}|x(t)|^2 \diff t}=\alpha^2.
\end{equation*}
Similarly, the effective bandwidth $W$ is the value such that 
\begin{equation*}
\frac{\int_{f_0-W/2}^{f_0+W/2}  |X(f)|^2 \diff f}{\int_{-\infty}^{\infty}|X(f)|^2 \diff f}=\beta^2.
\end{equation*}
We  transpose these formulas into the graph domain as follows. Given a vertex set $\S$ and a frequency set $\F$, using (\ref{D}) and (\ref{lowpass_operator}), the vectors $\mD \bx$ and $\mB \bx$ denote, respectively, the projection of $\bx$ onto the vertex set $\S$ and onto the frequency set $\F$. Then, we denote by $\alpha^2$ and $\beta^2$ the percentage of energy falling within the sets $\S$ and $\F$, respectively, as
\begin{equation}
\label{alpha-beta}
\frac{\|\mD \bx\|_2^2}{\| \bx\|_2^2}=\alpha^2;\,\, \frac{\|\mB \bx\|_2^2}{\| \bx\|_2^2}=\beta^2.
\end{equation}
In this paper, we find the region of all admissible pairs $(\alpha, \beta)$, by generalizing \cite{landau1961prolate} to the discrete case. More specifically, we express the boundaries of the admissible region in closed form and illustrate which are the signals that attain all the points of the admissible region.
It is worth noticing that, in (\ref{alpha-beta}), the graph topology is captured by the matrix $\mU$, present in the definition of the GFT in (\ref{eq::gft}), which appears inside the operator $\mB$. The theory presented in this paper is valid for any {\it unitary} mapping from some discrete space to its dual. 
\subsubsection{Sampling}
Building on the construction of a basis of maximally concentrated signals in the graph/dual domain, we express the conditions for recovering a band-limited signal from a subset of its values in terms of the properties of this basis. These conditions are equivalent to the conditions derived in \cite{pesenson2008sampling}, \cite{narang2013signal}, \cite{chen2015discrete}, and \cite{wang2014local}. The novelty here is that our formulation shows a direct link between sampling theory and uncertainty principle. It is shown that the unique recovery of any signal from $\B$, requires that there should be no nontrivial signal from $\B$ that is perfectly localized on $\oline{\S}$, i.e. one needs $\B_\F \cap \D_{\oline{\S}}$ to be empty. There may be various choices of $\S$ satisfying this requirement, but each choice may significantly affect the stability of the recovery algorithm, so that  selecting the sampling set $\S$ is a crucial step. Building on this idea, we propose several signal recovery algorithms and sampling strategies aimed to find an optimal sampling set $\S$. In addition, we propose a frame-based reconstruction method that fits perfectly into the given sampling framework, as it relies on the properties of the projectors $\mB$ and $\mD$.

Finally, we compare our algorithms with the methods proposed in \cite{chen2015discrete}, \cite{anis2014towards} and with the benchmark resulting from the solution of a combinatorial problem (only for small size networks, where the combinatorial search is still manageable). The comparison is carried out over a class of random graphs, namely the scale-free graphs, which are known for modeling many real world situations, see, e.g.  \cite{albert2002statistical}, \cite{Newman}, and our techniques exhibit advantages in terms of Mean Square Error (MSE) and show performance very close to the optimal theoretical bound. We also show an example of selection of the sampling set for a real network, namely the  IEEE 118 Bus Test Case, representing a portion of the American Electric Power System. 
\subsubsection{Signal recovery in case of strong impulsive noise}
Motivated by a sensor network scenario where some sensors may be damaged, 
we show under what conditions the recovery of a band-limited signal can be unaffected by some sort of impulsive noise affecting a subset of nodes, using $\ell_1$-norm minimization. Interestingly, we show that also this problem is inherently associated to the localization properties of projectors onto the graph and its dual domain.
The rest of the paper is organized as follows. In Section  \ref{Localization} we derive the localization properties of graph signals, illustrating as a particular case the conditions enabling perfect localization in {\it both} vertex and frequency domains. Building on these tools, in Section 
\ref{Uncertainty principle} we derive an uncertainty principle for graph signals and, in Section \ref{sec::Sampling},
we derive the necessary and sufficient conditions for recovering band-limited graph signal from its samples and propose alternative 
recovery algorithms. In Section \ref{sec::Reconstruction from noisy observations} we analyze the effect of observation noise on signal recovery and, finally, in Section \ref{Sampling strategies} we propose and compare several sampling strategies.

\section{Localization properties}
\label{Localization} 
Scope of this section is to derive the class of signals that are maximally concentrated 
over given subsets $\S$ and $\F$ in vertex and frequency domains. 
We say that a vector $\bx$ is perfectly localized over the subset $\S \subseteq \V$ if
\begin{equation}
\label{Dx=x}
\mD\, \bx=\bx,
\end{equation}
with $\mD$ defined as in (\ref{D}). 
Similarly, a vector $\bx$ is perfectly localized over the frequency set $\F \subseteq \V^*$ if
\begin{equation}
\label{Bx=x}
\mB\, \bx=\bx,
\end{equation}
with $\mB$ given in (\ref{lowpass_operator}).
Differently from continuous-time signals, a graph signal can be perfectly localized in {\it both} vertex and frequency domains. 
This is stated in the following theorem.
\begin{theorem} 
\label{theorem_unit_eigenvalue}
There exists a non trivial vector $\bx$, perfectly localized over both vertex set $\S$ and frequency set $\F$ (i.e. $\bx \in \B_{\F} \cap \D_{\S}$) if and only if the operator $\mB \mD \mB$ (or $\mD \mB \mD$) has an eigenvalue equal to one; in such a case, $\bx$ is an eigenvector associated to the unit eigenvalue.
\end{theorem}
\begin{proof}
Let us start proving that, if a vector $\bx$ is perfectly localized in both vertex and frequency domains, then it must be an eigenvector of $\mB \mD \mB$ associated to a unit eigenvalue. Indeed, by repeated applications of   (\ref{Dx=x}) and (\ref{Bx=x}), it follows
\begin{equation}
\label{BDBx = BDx}
\mB \mD \mB \bx = \mB \mD \bx = \mB \bx = \bx.
\end{equation}
This proves the first part. Now, let us prove that, if $\bx$ is an eigenvector of $\mB \mD \mB$ associated to a unit eigenvalue, then $\bx$ must satisfy  (\ref{Dx=x}) and (\ref{Bx=x}). Indeed, starting from
\begin{equation}
\label{BDBx=x}
\mB \mD \mB \bx = \bx
\end{equation}
and multiplying from the left side by $\mB$, taking into account that $\mB^2=\mB$, we get
\begin{equation}
\label{BDBx=Bx}
\mB \mD \mB \bx = \mB \bx
\end{equation}
Equating (\ref{BDBx=x}) to (\ref{BDBx=Bx}), we get
\begin{equation}
\label{eq::bandlim_x}
\mB \bx= \bx,
\end{equation}
which implies that $\bx$ is perfectly localized in the frequency domain. Now, using (\ref{eq::bandlim_x}) and the Rayleigh-Ritz theorem, we can also write
\begin{equation}
1=\max_{\bx} \frac{\bx^* \mB \mD \mB \bx}{\bx^* \bx}=\max_{\bx} \frac{\bx^* \mD \bx}{\bx^* \bx}.
\end{equation}
This shows that $\bx$ satisfies also (\ref{Dx=x}), i.e., $\bx$ is also perfectly localized in the vertex domain.
\end{proof}
Equivalently, the perfect localization properties can be expressed in terms of the operators $\mB \mD$ and $\mD \mB$. First of all, we prove the following lemma.
\begin{lemma}
\label{theorem::sing_val_bd_db}
The operators $\mB \mD$ and $\mD \mB$ have the same singular values, i.e. $\sigma_i(\mB \mD) = \sigma_i(\mD \mB), \ i = 1, \ldots, N$.
\end{lemma}
\begin{proof}
Since both matrices  $\mB$ and $\mD$ are Hermitian, $(\mB \mD)^* = \mD \mB$. But the singular values of a matrix coincide with the singular values of its Hermitian conjugate. 
\end{proof}
Combining Lemma \ref{theorem::sing_val_bd_db} and (\ref{BDBx = BDx}), perfect localization onto the sets $\S$ and $\F$ can be achieved if and only if
\begin{equation}
\label{|BD|=1=|DB|}
\|\mB \mD\|_2 = \|\mD \mB\|_2 = 1.
\end{equation}
As mentioned in Theorem \ref{theorem_unit_eigenvalue}, the vectors perfectly localized in both vertex and frequency domains must belong to the intersection set $\B \cap \D$, which is non-empty when the sum of dimensions of $\B$ and $\D$ is greater than the dimension of the ambient space of dimension $N$.
Hence, a sufficient condition for the existence of perfectly localized vectors in both vertex and frequency domains is
\begin{equation}
\label{Perfect_localization}
|\S|+|\F| > N.
\end{equation}
Conversely, if $|\S|+|\F| \leq N$, there could still exist perfectly localized vectors, when condition (\ref{|BD|=1=|DB|}) is satisfied.

Typically, given two generic domains $\S$ and $\F$, we may have signals that are not perfectly concentrated in both domains. In such a case, it is worth finding the class of signals with limited support in one domain and maximally concentrated on the dual one. For example, we may search for the orthonormal set of perfectly band-limited signals, i.e. satisfying $\mB \bx = \bx$, which are maximally concentrated in a vertex domain $\S$. 
The set of such vectors $\{\bpsi_i\}$ is constructed as the solution of the following iterative optimization problem, for $i=1, \ldots, N$:
\begin{equation}
\label{slep_func:max_problem}
\begin{aligned}
\bpsi_i &&= \  & \underset{\bpsi_i}{\arg \max} \ \| \mathbf{D}  \bpsi_i \|_2 \\
&&& \text{s.t. }  \| \bpsi_i\|_2 = 1,\\ 
&&& \mathbf{B} \bpsi_i = \bpsi_i,\\
&&&  \langle \bpsi_i, \bpsi_j \rangle = 0, \ \ j=1, \ldots,  i-1, \ \ {\rm if}\,\, i>1.
\end{aligned}
\end{equation}
In particular, $\bpsi_1$ is the band-limited signal with the highest energy concentration on $\S$;  $\bpsi_2$ is the band-limited signal, orthogonal to $\bpsi_1$, which is maximally concentrated on $\S$, and so on. The vectors $\{ \bpsi_i \}$ are the counterpart of the prolate spheroidal wave functions introduced by Slepian and Pollack for continuous-time signals \cite{Slepian:1961:PSW}. The solution of the above optimization problem is given by the following theorem.
\begin{theorem}
\label{theorem::max_concentrated_vectors}
The set of orthonormal $\F$-band-limited vectors $\{ \bpsi_i \}_{i=1, \ldots, K}$, with $K := \rank \mB$,  maximally concentrated over a vertex set $\S$, is given by the eigenvectors of the $\mB \mD \mB$ operator, i.e.
\begin{equation}
\label{slep_func:bdb_eigendecomposition}
\mathbf{BDB} \bpsi_i = \lambda_i \bpsi_i,
\end{equation}
with $\lambda_1\ge \lambda_2 \ge \ldots \ge \lambda_K$.
Furthermore, these vectors are orthogonal over the set $\S$, i.e.
\begin{equation}
\label{eq:orthogonality_on_sampling_set}
\langle \bpsi_i, \mathbf{D} \bpsi_j \rangle = \lambda_j \delta_{ij},
\end{equation}
where $\delta_{ij}$ is the Kronecker symbol.
\end{theorem}
\begin{proof}
Substituting the band-limiting constraint within the objective function in (\ref{slep_func:max_problem}), we get
\begin{equation}
\label{slep_func:max_problem_restated}
\begin{aligned}
\bpsi_i &&= \  & \underset{\bpsi_i}{\arg \max} \ \| \mathbf{D} \mathbf{B}  \bpsi_i \|_2 \\
&&& \text{s.t.} \ \ \| \bpsi_i\|_2 = 1, \ \ \langle \bpsi_i, \bpsi_j \rangle = 0, \ \ j\neq i.
\end{aligned}
\end{equation}
Using Rayleigh-Ritz theorem, the solutions of (\ref{slep_func:max_problem_restated})  are the eigenvectors of $\left(\mathbf{DB}\right)^* \mathbf{DB} = \mathbf{BDB}$, i.e. the solutions of (\ref{slep_func:bdb_eigendecomposition}). This proves the first part of the theorem. The second part is proven by noting that, using $\mB \bpsi_i=\bpsi_i$ and $\mB^*=\mB$, we obtain $\langle \bpsi_i, \mB \mD \mB \bpsi_j \rangle = \langle \bpsi_i, \mathbf{D} \bpsi_j \rangle = \lambda_j \delta_{ij}$.
\end{proof}
The above theorem provides a set of perfectly $\F$-band-limited vectors that are maximally concentrated over a vertex domain. The same procedure can of course be applied to identify the class of orthonormal vectors perfectly localized in the graph domain and maximally concentrated in the frequency domain, simply exchanging the role of $\mB$ and $\mD$, and thus referring to the eigenvectors of $\mD \mB \mD$. 

\section{Uncertainty principle}
\label{Uncertainty principle}
Quite recently, the uncertainty principle was extended to signals on graphs in {\cite{agaskar2013spectral}} by following an approach based on the transposition of the definitions of time and frequency spreads given by ({\ref{Delta_T}}) and ({\ref{Delta_F}}) to graph signals, as indicated in (\ref{delta_t_AgaskarLu}) and (\ref{delta_f_AgaskarLu}). However, as mentioned in the Introduction, the computation of spreads based on second order moments implies a definition of distance over a graph. Although the definition of distance used in  {\cite{agaskar2013spectral}}, based on the shortest path between two vertices, is perfectly reasonable, there are alternative definitions of distance over a graph, such as the resistance distance \cite{KleinRandic93} or the diffusion distance \cite{CoifmanMaggioni06}. To remove any potential ambiguity associated to the definition of distance over a graph, taking inspiration by the seminal works of  Slepian, Landau and Pollack {\cite{Slepian:1961:PSW}}, {\cite{landau1961prolate}}, in this paper we resort to a definition of spread in the graph and frequency domain that does not imply any definition of distance. 
More specifically, given a pair of vertex set $\S$ and frequency set $\F$, denoting by $\alpha^2$ and $\beta^2$ the percentage of energy falling within the sets $\S$ and $\F$, respectively, as defined in (\ref{alpha-beta}), our goal is to establish the trade-off between $\alpha$ and $\beta$ and
find out the signals able to attain all admissible pairs. The resulting uncertainty principle is stated in the following theorem. \\

\begin{theorem}
\label{theorem::Uncertainty principle}
There exists a vector $\bx$ such that $\|\bx\|_2 = 1$, $\| \mD \bx \|_2 = \alpha$, $\| \mB \bx \|_2 = \beta $ if and only if $\left( \alpha, \beta \right) \in \Gamma$, where 
\begin{align}
\label{eq::uncertainty_region_Gamma}
\Gamma &= \left\{ \left( \alpha , \beta \right) : \right. \nonumber \\
& \cos^{-1} \alpha + \cos^{-1} \beta \geq \cos^{-1} \sigma_{max} \left( \mB \mD \right), \nonumber \\
& \cos^{-1} \sqrt{1 - \alpha^2} + \cos^{-1} \beta \geq \cos^{-1} \sigma_{max} \left( \mB \mDc \right), \\
& \cos^{-1} \alpha + \cos^{-1} \sqrt{ 1 - \beta^2} \geq \cos^{-1} \sigma_{max} \left( \mBc \mD \right), \nonumber \\
& \left. \cos^{-1} \sqrt{ 1- \alpha^2 } + \cos^{-1} \sqrt{1 - \beta^2} \geq \cos^{-1} \sigma_{max} \left( \mBc \mDc \right) \right\}. \nonumber 
\end{align}
\end{theorem}
\begin{proof} 
The proof is reported in Appendix A. 
\end{proof}

\begin{figure}
\vspace{-2cm}
\begin{minipage}[b]{1.0\linewidth}
 \centering
 \centerline{\includegraphics{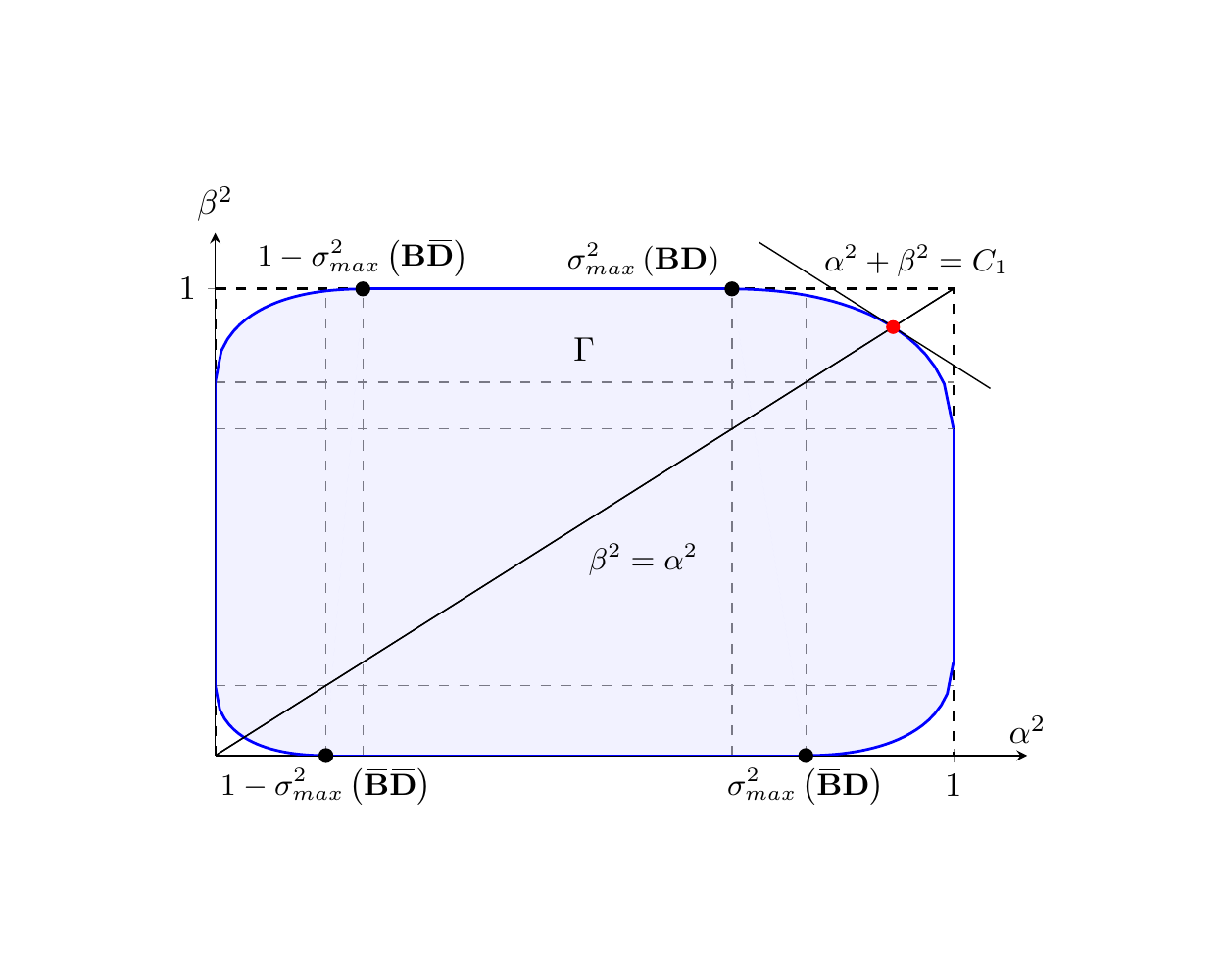}}
\end{minipage}
\vspace*{-2cm}
\caption{\small Admissible region $\Gamma$ of unit norm signals $\bx$ with $\| \mD \bx \|_2=\alpha$ and $\| \mB \bx \|_2=\beta$.}
\vspace*{-0.5cm}
\label{fig:uncertainty}
\end{figure}

An illustrative example of admissible region $\Gamma$ is reported in Fig. \ref{fig:uncertainty}. A few remarks about the border of the region $\Gamma$ are of interest. First of all, if we take the equality signs in the inequalities given in (\ref{eq::uncertainty_region_Gamma}), we get the equations describing the curves appearing at the four corners sketched in Fig. \ref{fig:uncertainty}, namely  upper right, upper left, bottom right and bottom left, respectively. The upper right corner of $\Gamma$, in particular, specifies the pairs $(\alpha, \beta)$ that yield the maximum concentration over both graph and dual domains. This curve has equation
\begin{equation}
\label{eq::beta alpha curve}
\cos^{-1} \alpha + \cos^{-1} \beta = \cos^{-1} \sigma_{max}(\mathbf{B}\mathbf{D}).
\end{equation}
Solving (\ref{eq::beta alpha curve}) with respect to $\beta$, and setting $\sigma^2_{max} := \sigma^2_{max}(\mathbf{B}\mathbf{D})$, we get
\begin{equation}
\label{beta vs alpha}
\beta = \alpha\, \sigma_{max} + \sqrt{(1-\alpha^2)(1-\sigma_{max}^2)}.
\end{equation}
Typically, for any given subset of nodes $\S$, as the cardinality of $\F$ increases, this upper curve gets closer and closer to the upper right corner. The curve collapses onto a point, namely the upper right corner, when the sets $\S$ and $\F$ give rise to projectors $\mD$ and $\mB$ that satisfy the perfect localization conditions in (\ref{|BD|=1=|DB|}). 
In general, any of the four curves at the corners of region $\Gamma$ in Fig. \ref{fig:uncertainty} may collapse onto the corresponding corner, whenever the conditions for perfect localization of the corresponding operator hold true.

In particular, if we are interested in the allocation of energy within the sets $\S$ and $\F$ that maximizes, for example, the sum of the (relative) energies $\alpha^2+\beta^2$ falling in the vertex and frequency domains, the result is given by the intersection of the upper right curve, i.e. (\ref{beta vs alpha}), with the line $\alpha^2+\beta^2={\rm const}$. Given the symmetry of the curve (\ref{eq::beta alpha curve}), the result is achieved by setting $\alpha=\beta$, which yields
\begin{equation}
\alpha^2=\frac 1 2 (1+\sigma_{max}).
\end{equation}
Using the derivations reported in Appendix A, the corresponding function $\bff'$ may be written in closed form as
\begin{equation}
\label{eq::opt_vectro_closed_form}
\bff' = \frac{\bpsi_1 - \mD \bpsi_1}{\sqrt{2\left(1 + \sigma_{max} \right)}} + \sqrt{\frac{1+\sigma_{max}}{2\sigma^2_{max}}} \mD \bpsi_1,
\end{equation}
where $\bpsi_1$ is the eigenvector of $\mB \mD \mB$ corresponding to $\sigma_{max}^2$. 
More generally, we can find all the vectors whose vertex and spectral energy concentrations lie on the border of 
the uncertainty region $\Gamma$ and construct the corresponding sets of orthonormal vectors by considering the following optimization problem
\begin{equation}
\label{eq::max_sum_energy}
\begin{aligned}
\bff_i &&= \  & \underset{\bff_i: \ \| \bff_i\|_2 = 1}{\arg \max} \ \gamma \| \mB  \bff_i \|^2_2 + (1-\gamma) \| \mD \bff_i \|^2_2  \\
&&& \text{s.t.} \ \  \langle \bff_i, \bff_j \rangle = 0, \ \ j\neq i,
\end{aligned}
\end{equation}
where the parameter $\gamma$, with $0< \gamma < 1$, controls the relative energy concentration in the vertex and frequency domains. The solution of this problem is given by the eigenvectors of the matrix $\gamma \mB + (1-\gamma)\mD$. In particular, it is interesting to notice, as detailed in Appendix B, that the first $K$ eigenvectors of this matrix, with $K={\rm rank}(\mB \mD)$, are related to the eigenvectors $\bpsi_i$ associated to the $K$ largest eigenvalues  of $\mB \mD \mB$ by the following relation 
\begin{equation}
\label{eq::extreme_energy_function general}
\bff_i = p_i \bpsi_i + q_i \mD \bpsi_i,
\end{equation}
where 
\begin{equation}
\label{eq::p_def general}
p_i = \sqrt{ \frac{1-\alpha_i^2}{1 - \sigma^2_i }},
\end{equation}
\begin{equation}
\label{eq::q_def general}
q_i = \frac{\alpha}{\sigma_i } - \sqrt{\frac{1-\alpha_i^2}{1 - \sigma_i^2}}
\end{equation}
with $\sigma_i := \sigma_i \left( \mB \mD \right)$, and
\begin{align}
\label{eq::alpha f1 general}
\alpha_i = \sqrt{\frac{1}{2} \left( \frac{2\gamma \left( \sigma^2_i -1 \right) + 1}{\sqrt{(1-2\gamma)^2-4\gamma (\gamma-1)\sigma^2_i}} +1 \right)}.
\end{align}
 
A numerical example is useful to grasp the advantages of tolerating some energy spill-over in representing a graph signal. The example is built as follows. We consider a random geometric graph composed of $100$ vertices, where a set of nodes is deployed randomly within a finite area and there is an edge between two nodes if their Euclidean distance is less than a given coverage radius $r_0$. To avoid problems with points close to the boundary of the deployment region, which would have statistics different from the internal nodes, we simulated a toroidal surface, so that all points are statistically equivalent in terms of graph properties, like degree, clustering, etc. Then, we picked a vertex $i_0$ at random and identify the set $\S$ as the ensemble of  nodes falling within a distance $R_0$ from $i_0$. Then we let $R_0$ to increase and, for each value of $R_0$, we evaluate the cardinality of $\S$ and we build $\F$ as the set of indices $\{1, 2, \ldots, k\}$ enumerating the first $k$ eigenvectors of the Laplacian matrix $\mL$, where $k$ is the minimum number such that the (relative) spill-over energy $1-\alpha^2$ $=1-\sigma^2_{max}$
is less than a prescribed value $\varepsilon^2$. In Fig. \ref{fig::F_vs_S} we plot $|\F|=k$ as a function of $|\S|$, for different values of  $\varepsilon^2$. The dashed line represents the case $\varepsilon^2=0$: This is the curve of equation $N=|\S|+|\F|$. The interesting result is that, as we allow for some spill-over energy, we can get a substantial reduction of the  ``bandwidth'' $|\F|$ necessary to contain a signal defined on a vertex set $\S$.  

\section{Sampling}
\label{sec::Sampling}
Given a signal $\bx \in \B$ defined on the vertices of a graph, let us denote by $\bxs \in \D$ the vector equal to $\bx$ on the subset $\S \subseteq \V$ and zero outside:
\begin{equation}
\label{r=Ds}
\bxs := \mathbf{D} \bx.
\end{equation}
The necessary and sufficient condition for perfect recovery of $\bx$ from $\bxs$ is stated in the following theorem.
\begin{theorem}
\label{theorem::sampling theorem}
Given a sampled signal as in (\ref{r=Ds}), it is possible to recover $\bx \in \B$ from its samples $\bxs$,  for {\it any} $\bx \in \B$, 
if  and only if
\begin{equation}
\label{|DcB|<1}
\| \mB \mDc \|_2 < 1,
\end{equation}
i.e. if the matrix $\mB \mDc \mB$ does not have any eigenvector that is perfectly localized on $\oline{\S}$ and band-limited on $\F$.
\end{theorem}
\begin{proof}
We prove first that condition (\ref{|DcB|<1}) is sufficient for perfect recovery.  Let us denote by  $\mQ$ a matrix enabling the reconstruction of $\bx$ from $\bxs$ as $\mathbf{Q} \bxs$. 
If such a matrix exists, the corresponding reconstruction error is
\begin{equation}
\bx - \mQ \bxs = \bx - \mQ \left( \mI - \mDc \right) \bx=\bx - \mQ \left( \mI - \mDc \mB \right) \bx, 
\end{equation}
where, in the second equality, we exploited the band-limited nature of $\bx$.
This error can be made equal to zero by taking $\mQ=\left( \mI - \mDc \mB \right)^{-1}$. Hence, checking for the existence of $\mQ$ is equivalent to check if $\left( \mI - \mDc \mB \right)$ is invertible. This happens if (\ref{|DcB|<1}) holds true. Conversely, if 
$\|\mathbf{B}\oline{\mathbf{D}}\|_2=1$ and, equivalently,
$\|\oline{\mathbf{D}}\mathbf{B}\|_2=1$, from (\ref{|BD|=1=|DB|}) we know that there exist band-limited signals that are perfectly localized over $\oline{\S}$. This implies that, if we sample one of such signals over the set $\S$, we get only zero values and then it would be impossible to recover $\bx$ from those samples. This proves that condition (\ref{|DcB|<1}) is also necessary. 
\end{proof}
Theorem \ref{theorem::sampling theorem} suggests also a way to recover the original signal from its samples as $\left( \mI - \mDc \mB \right)^{-1}\bxs$. Alternative recovery strategies will be suggested later on. Before considering the recovery algorithms,  we note that, if $\bx \in \B$ then
\begin{equation}
\label{db_equality}
\left( \mI - \mDc \mB \right) \bx = \mD \mB \bx.
\end{equation}
The operator $\mD \mB$ is invertible, for any $\bx \in \B$, if the dimensionality of the image of $\mD \mB$ is equal to the $\rank \mB$, i.e.
\begin{equation}
\label{db_equal_b}
\rank \mD \mB = \rank \mB.
\end{equation}
This condition is then equivalent to the condition of Theorem \ref{theorem::sampling theorem}. In this case the singular vectors of $\mD \mB$ corresponding to non-zero singular values constitute a basis for $\B$. 
In general, both conditions (\ref{|DcB|<1}) and (\ref{db_equal_b}) are equivalent to the sampling theorem conditions derived, for example, in \cite{anis2014towards} or \cite{chen2015discrete}. The interesting remark here is that formulating the sampling conditions as in  (\ref{|DcB|<1}) highlights a strict {\it link between sampling theory and uncertainty principle}. In fact, if we look at the top-left corner of the admissible region in Fig. \ref{fig:uncertainty}, it is clear that if the signal is perfectly band-limited over a subset $\F$, then $\beta^2=1$. To enable signal recovery from a subset of samples $\S$,  we need to avoid the possibility that  $\alpha^2=0$, because this would make signal recovery impossible. From Fig. \ref{fig:uncertainty}, it is clear that this is possible only if $\sigma_{max} \left( \mB \mDc \right)<1$, i.e. if (\ref{|DcB|<1}) holds true, as stated in Theorem \ref{theorem::sampling theorem}. More generally, if we allow for some energy spill-over in the frequency domain, so that we take $\beta=\oline{\beta}<1$, to avoid the condition $\alpha^2=0$, we need to check that
$\sigma_{max}^2 \left( \mB \mDc \right)<1-\oline{\beta}^2$. Having conditions in this form is indeed useful to devise possible sampling strategies, as it suggests to take $\sigma_{max} \left( \mB \mDc \right)$ as a possible objective function to be minimized. This topic will be addressed more closely in Section \ref{Sampling strategies}, when dealing specifically with sampling strategies. 

\begin{figure}[t]
\centering
\vspace{-1cm}
\includegraphics[width=\linewidth,height=\textheight,keepaspectratio]{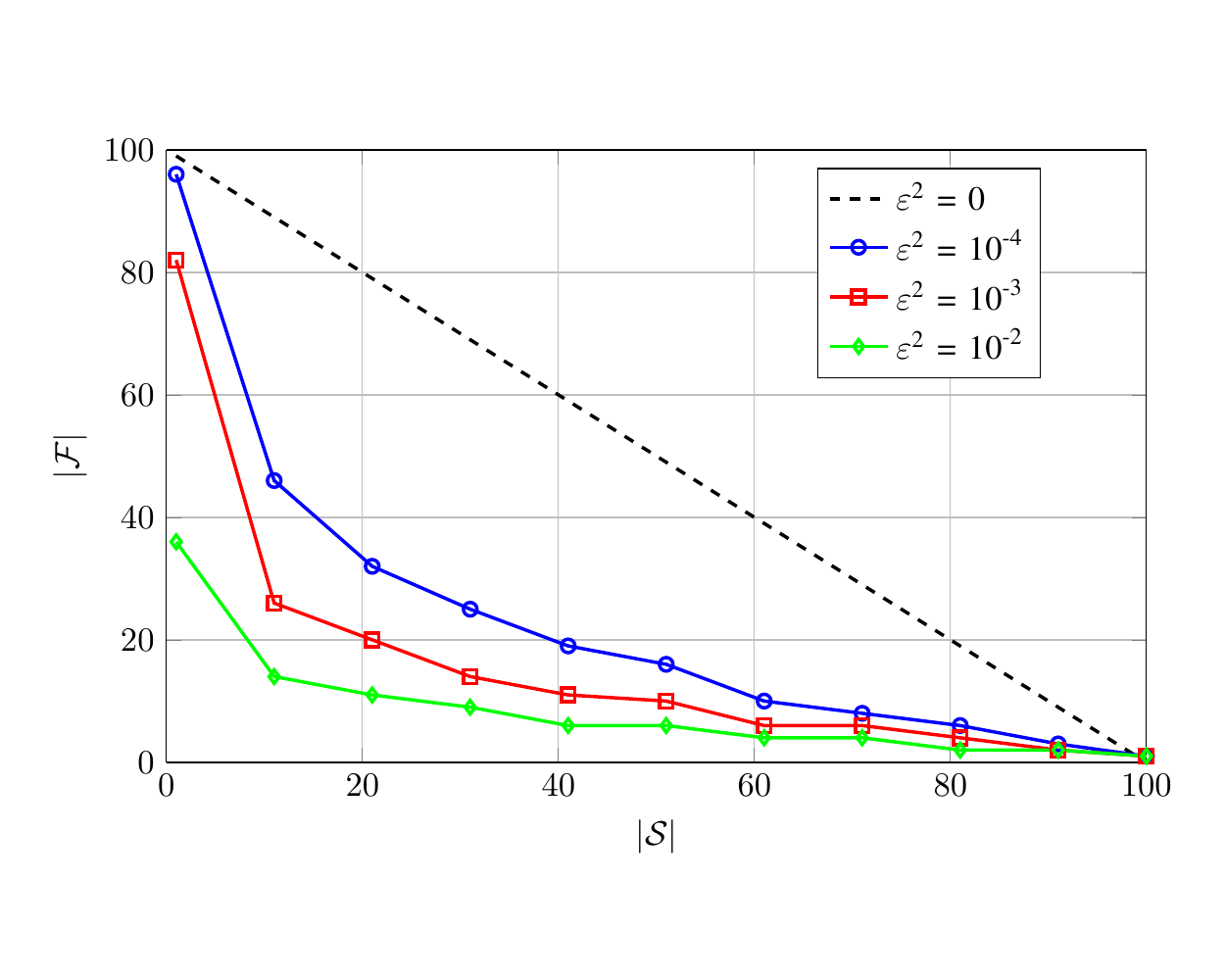}
\vspace{-1.5cm}
\caption{\small Relation between the dimensions of the support over the vertex set $\S$ and the frequency domain $\F$ guaranteeing 
a spill-over energy $\varepsilon^2$.}\label{fig::F_vs_S}
\vspace*{-0.5cm}
\end{figure}

The conceptual link between sampling theory and localization properties in the graph and dual domains is also useful to derive a signal recovery algorithm that builds on the properties of maximally concentrated signals described in Section \ref{Localization}, as established in the following.

\begin{theorem}
If condition (\ref{|DcB|<1}) of the sampling theorem holds true, then any band-limited signal $\bx \in \B$ can be reconstructed from its sampled version $\bx_{\scriptsize \S} \in \D$ by the following formula
\begin{equation}
\label{eq::sampling_theorem_formula}
\bx = \sum_{i=1}^{\scriptsize \abs{\F}} \frac{1}{\sigma_i^2} \langle \bxs, \bpsi_i \rangle \bpsi_i,
\end{equation}
where $\left\{ \bpsi_i \right\}_{i = 1 .. K}$ and $\left\{ \sigma^2_i \right\}_{i = 1..K}$ with $K=|\F|$, are the eigenvectors and eigenvalues of $\mB \mD \mB$.
\end{theorem}
\begin{proof}
For band-limited projection of any $\bg$, we can write 
\begin{equation}
\mB \bg = \sum_{i = 1}^{K} \langle \mB \bg, \bpsi_i \rangle \bpsi_i.
\end{equation}
Because of (\ref{|DcB|<1}), there is no  band-limited vector in $\B$ perfectly localized on $\oline{\S}$. Hence, all the eigenvectors from $\ker (\mB \mD \mB)$ belong to $\oline{\B}$, so  that $K = |\F|$.
Setting $\bx = \mB \bg$, since $\mB \mD \mB \bpsi_i=\sigma_i^2 \bpsi_i$ with $\sigma_i \neq 0$ for $i\in \F$, we can then write
\begin{equation}
\bx = \sum_{i = 1}^{\scriptsize \abs{\F}} \langle \bx, \frac{1}{\sigma^2_i} \mB \mD \mB \bpsi_i \rangle \bpsi_i =
\sum_{i = 1}^{\scriptsize \abs{\F}} \frac{1}{\sigma^2_i}  \langle \mD \bx, \bpsi_i \rangle \bpsi_i,
\end{equation}
where we have used the property that the operators $\mD$ and $\mB$ are self-adjoint and the eigenvectors $\left\{ \bpsi_i \right \}_{\scriptsize i=1,\dots,\abs{\F}}$ are band-limited. 
\end{proof}

\begin{figure}[t]
\vspace{-1cm}
\centering
\includegraphics[width=\linewidth]{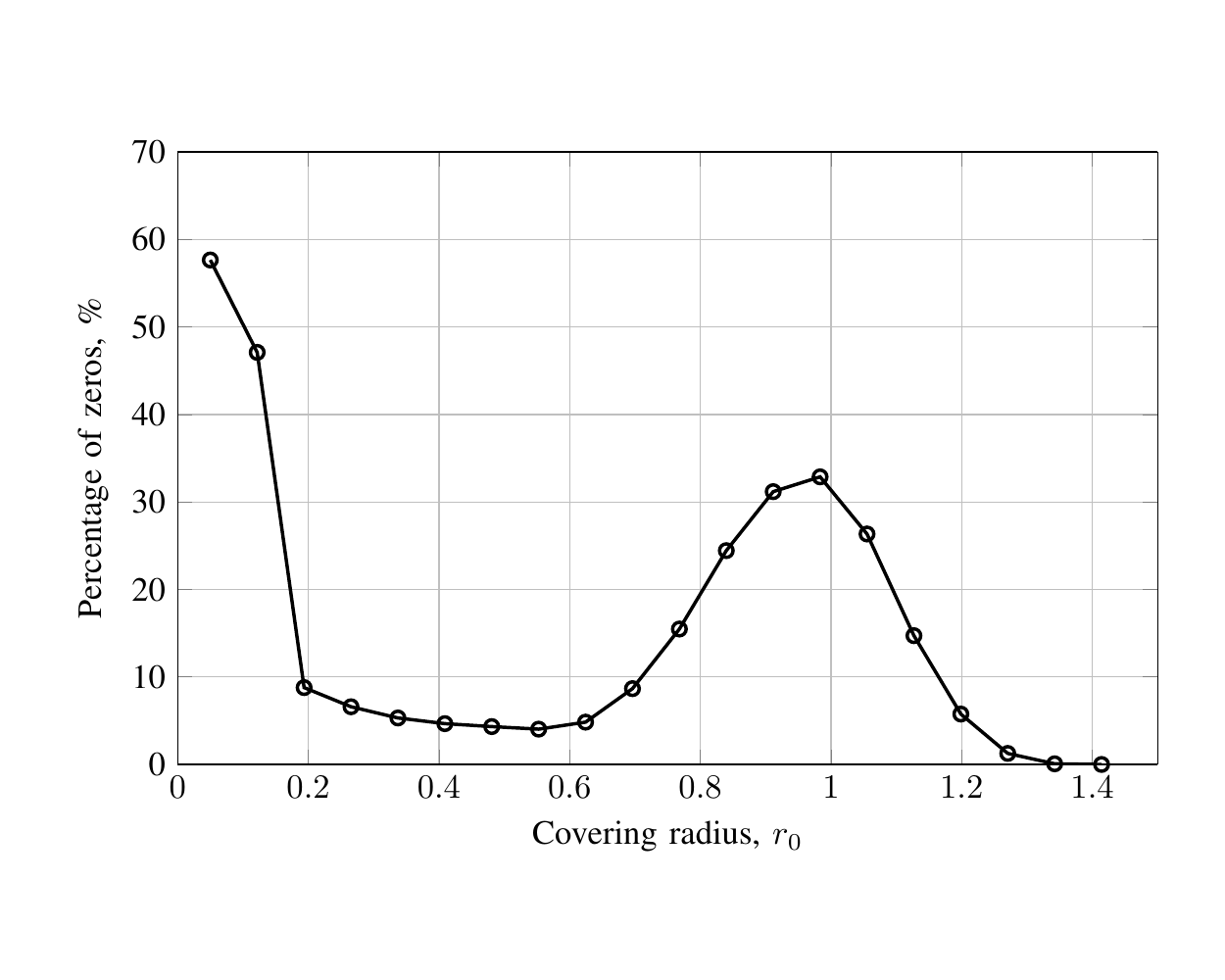}
\vspace{-1.5cm}
\caption{\small Percentage of vanishing entries of the Laplacian eigenvectors of a RGG vs. coverage radius $r_0$.}
\vspace{-0.3cm}
\label{zeros}
\end{figure}
Let us study now the implications of condition (\ref{|DcB|<1}) of Theorem \ref{theorem::sampling theorem} on the sampling strategy. To fulfill (\ref{|DcB|<1}), we need to guarantee that there exist no band-limited signals, i.e. $\mB \bx =\bx$, such that $\mB\mDc\bx=\bx$. 
To make (\ref{|DcB|<1}) hold true, we must then ensure that $\mB \mDc \bx \neq \bx$ or, equivalently, recalling Lemma \ref{theorem::sing_val_bd_db},  $ \mDc \mB \bx \neq \bx$. Since
\begin{equation}
\mB \bx = \bx = \mD\mB \bx+\mDc \mB \bx,
\end{equation}
we need to guarantee that  $\mD\mB \bx \neq\b0$. To this purpose, let us define the $|\S| \times |\F|$ matrix $\mG$ as
\begin{equation}
\mG=
\left(
\begin{array}{llll}
u_{i_1}(j_1) & u_{i_2}(j_1) & \cdots & u_{i_{\tiny|\F|}}(j_1)\nonumber\\
\vdots & \vdots & \vdots & \vdots \nonumber\\
u_{i_1}(j_{\tiny|\S|}) & u_{i_2}(j_{\tiny|\S|}) & \cdots & u_{i_{\tiny|\F|}}(j_{\tiny|\S|})
\end{array}
\right)
\end{equation}
whose $\ell$-th column is the eigenvector of index $i_{\ell}$ of the Laplacian matrix (or any orthonormal set of basis vectors), sampled at the positions indicated by the indices $j_1, \ldots, j_{\tiny|\S|}$. Condition (\ref{|DcB|<1}) is equivalent to require $\mG$ to be full column rank. 

Indeed, the eigenvectors of a graph Laplacian may contain several vanishing elements, so that matrix $\mG$ may easily loose rank.  
As an extreme case, if the graph is not connected, the vertices can be labeled so that the Laplacian (adjacency) matrix can be written as a block diagonal matrix, with a number of blocks equal to the number of connected components. Correspondingly, each eigenvector of $\mL$ can be expressed as a vector having all zero elements, except the entries corresponding to the connected component, which that eigenvector is associated to. This implies that, if there are no samples over the vertices corresponding to the non-null entries of the eigenvectors with index included in $\F$, $\mG$ looses rank. In principle, a signal defined over a disconnected graph can still be reconstructed from its samples, but only provided that the number of samples belonging to each connected component is at least equal to the number of eigenvectors with indices in ${\F}$ associated to that component. More generally, even if the graph is connected, there may easily occur situations where matrix $\mG$ is not rank-deficient, but it is ill-conditioned, depending on graph topology and samples' location. 

A numerical example is useful to grasp the criticality associated to sampling. In Fig. \ref{zeros}, we report the percentage of vanishing ($< 1.e-10$) entries of the Laplacian eigenvectors of a random geometric graph (RGG), composed of $N=100$ nodes uniformly distributed over a unit square, vs. coverage radius $r_0$. 
The results shown in Fig. \ref{zeros} are obtained by averaging over $100$ independent realizations of RGG's.
The behavior of the curve can be explained as follows. The value of $r_0$ that ensures the graph connectivity with high probability is approximately $r_0 \approx \sqrt{\log(N)/N} \approx 0.2$. This means that, for $r_0<0.2$, there are disconnected components and this explains the high number of zeros. For $0.2 < r_0<0.6$ the graph is typically composed of a giant component and the number of vanishing entries is relatively low. Then, for $0.6<r_0<1.2$ the graph is connected with very high probability, but there appear clusters and the eigenvectors of the Laplacian may have several entries close to zeros as a way to evidence the presence of clusters. Finally, for $r_0>1.2$ the graph tends to be fully connected and there are no zero entries anymore. 
We can see from Fig. \ref{zeros} that the percentage of vanishing entries can be significant. 
This implies that the location of samples plays a key role in the performance of the reconstruction algorithm. For this reason, In Section \ref{Sampling strategies} we will suggest and compare a few alternative sampling strategies satisfying different optimization criteria.

{\it Frame-based reconstruction}: The problem of sampling on graphs using frames for the space $\B$ was initially studied by \cite{pesenson2008sampling}, \cite{pesenson2010sampling}, where the conditions for the existence of such frames were derived. Here we approach the problem using the above developed theory of maximally vertex-frequency concentrated signals on graph. First we provide some basic definitions of the frame theory \cite{duffin1952class}.
\begin{definition}
A set of elements $\left\{ \bg_i \right\}_{ i \in \I}$, is a frame for the Hilbert space $\H$, if for all $\bff \in \H$ there exist constants $0 < A \leq B < \infty$ such that
\begin{equation}
\label{frame_def}
A \| \bff \|^2_2 \leq \sum_{i \in \I} | \langle \bff, \bg_i \rangle |^2 \leq B \| \bff \|^2_2.
\end{equation}
\begin{definition}
Given a frame $\left\{ \bg_i \right\}_{ i \in \I}$, the linear operator $\mT: \H \rightarrow \H$ defined as
\begin{equation}
\label{frame_operator}
\mT \bff = \sum_{i \in \I} \langle \bff, \bg_i \rangle \bg_i
\end{equation}
is called the frame operator. 
\end{definition}
Constants $A$ and $B$ are called frame bounds, while the largest $A$ and the smallest $B$ are called the tightest frame bounds.
It is useful to note that condition (\ref{frame_def}) guarantees the frame operator $\mT$ to be bounded and invertible.
\end{definition}
Now, introducing the canonical basis vector $\bdelta_u$, with $u\in {\cal V}$, i.e. having all zero entries except the $u$-th entry equal to $1$, we investigate under what conditions a set of vectors $\left\{ \mB \bdelta_u \right\}_{\scriptsize u \in \S} $ constitutes a frame for $\B$. The frame operator in this case is
\begin{equation}
\label{dirac_frame_expansion}
\mT_\delta \bff = \sum_{\scriptsize u \in \S} \langle \bff, \mB \bdelta_u \rangle \mB \bdelta_u= \sum_{\scriptsize u \in \S} f(u) \mB \bdelta_u.
\end{equation}

First, we observe that the frame operator $\mT_\delta$, as defined in (\ref{dirac_frame_expansion}), may be also expressed as
\begin{equation}
\label{t_bd}
\mT_\delta = \mB \mD \mD \mB = \mB \mD \mB.
\end{equation}
Operator $\mT_\delta$ has a spectral norm $\| \mT_\delta \|_2$ equal to $\sigma^2_{max} \left( \mB \mD \right)$. Hence,  to guarantee that $\left\{ \mB \bdelta_u \right\}_{\scriptsize u \in \S} $ is a frame, it is sufficient to check when $\mT_\delta$ is invertible, for any $\bff \in \B$. The operator $\mB \mD$, on its turn, is invertible for any $\bff \in \B$ if its singular vectors, not belonging to its kernel, constitute a basis for the $\abs{\F}$-dimensional space $\B$, or, formally, if and only if
\begin{equation}
\label{bd_equal_b}
\rank \mB \mD \mB = \rank \mB.
\end{equation}
Taking into account (\ref{db_equal_b}) and Lemma \ref{theorem::sing_val_bd_db}, we conclude that the condition for a frame-based reconstruction based on a canonical-vector frames coincides with the condition of Theorem \ref{theorem::sampling theorem}.

In general, however, the reconstruction based on the canonical-vector frame may be non robust in the presence of observation noise. For this reason, we generalize the sampling frame operator $\mT_{\delta}$ by introducing the operator $\mT_Y$ as
\begin{equation}
\label{eq::Y_def}
\mT_Y \bff = \mB \mY \mD \mB \bff = \sum_{\scriptsize u \in \S} f(u) \by_u,
\end{equation}
where $\mY$ is a bounded matrix whose columns $\by_i$, without loss of generality, can be taken belonging to $\B$, i.e. $\mB \by_i=\by_i$, so that the image of $\mY$ is also $\F$-band-limited. Let us consider now the reconstruction of $\bff \in \B$ from its samples on $\S$, based on $\mT_Y$. This requires  checking under what conditions the operator $\mT_Y$ is bounded and invertible.  Since the columns of $\mY \mD$ corresponding to indices that do not belong to the set $\S$ are null, we can limit our attention to matrices $\mY$ that are invariant to the right-side multiplication by $\mD$, i.e. $\mY \mD = \mY$. Finally, we arrive at the following sampling theorem.
\begin{theorem}
\label{theorem::frame_theorem}
Let $\F \subseteq \V^*$ be the set of frequencies and $\S \subseteq \V$ be the sampling set of vertices and let $\mY : \B_{\F} \rightarrow \mathbb{C}^N$ be an arbitrary bounded operator, then $\left\{ \mB\by_i \right\}_{\scriptsize i \in \S}$ is a frame for $\B_{\F}$ if and only if
\begin{equation}
\label{bad_equal_b}
\rank \mB \mY \mD \mB = \rank \mB.
\end{equation}
\end{theorem}
\begin{proof}
The proof follows directly from the invertibility conditions for the operator $\mB \mY \mD \mB$. 
\end{proof}
The tightest frame bounds, according to the Rayleigh-Ritz theorem, are defined by the minimum and maximum singular values of $\mB \mY \mD \mB$
\begin{equation}
\label{A_frame_bounds}
\sigma_{min}  \| \bff \|^2_2 \leq \sum_{u \in S} | \langle \bff, \by_u \rangle |^2 \leq \sigma_{max} \| \bff \|^2_2,
\end{equation}
which is valid for every $\bff \in \B$. 
As an example of matrix $\mY$, encompassing the approaches  proposed in \cite{wang2014local} and  \cite{pesenson2015sampling}, we have the following frame operator
\begin{equation}
\label{eq::limited_propagation}
\mT_1 \bff = \mB \mY_1 \mD \mB \bff =  \sum_{u \in \S } f(u) \mB \bdelta_{\scriptsize \N (u)},
\end{equation}
where $\bdelta_{\scriptsize \N (u)}$ is the indicator function of set $\N(u)$, defined as $\delta_{\scriptsize \N(u)} (v) = 
1$, if $ \ v \in \N(u)$, and zero otherwise. 
In this case, the graph signal is supposed to be sampled sparsely in such a way that around each sampled vertex there is a non-empty neighborhood $\N(u)$ of vertices that altogether could cover the whole graph. However, this choice is not necessarily the best one. In Section 
\ref{Sampling strategies} we will provide numerical results showing how the generalized frame-based approach can yield better performance results in the presence of observation noise. 

\section{Reconstruction from noisy observations}
\label{sec::Reconstruction from noisy observations}
Let us consider now the reconstruction of band-limited signals from noisy samples, where
the observation model is
\begin{equation}
\label{r=D(s+n)}
\br = \mD \left( \bs + \bn \right),
\end{equation}
where $\bn $ is a noise vector. Applying (\ref{eq::sampling_theorem_formula}) to $\br$, the reconstructed signal $\tilde{\bs}$ is
\begin{equation}
\label{eq::signal_plus_noise_expansion}
\tilde{\bs} = \sum_{i=1}^{\abs{\F}} \frac{1}{\sigma_i^2} \langle \mD \bs, \bpsi_i \rangle \bpsi_i + \sum_{i=1}^{\abs{\F}} \frac{1}{\sigma_i^2} \langle \mD \bn, \bpsi_i \rangle \bpsi_i.
\end{equation}
Exploiting the orthonormality of $\bpsi_i$, the mean square error is
\begin{align}
\label{eq::expected_value_noise}
MSE &= \mathbb{E}\left\{\| \tilde{\bs}-\bs \|_2^2 \right\} =\mathbb{E}\left\{ \sum_{i=1}^{\abs{\F}} \frac{1}{\sigma_i^4} \abs{\langle \mD \bn, \bpsi_i \rangle }^2 \right\} \nonumber \\
&=  \sum_{i=1}^{\abs{\F}} \frac{1}{\sigma_i^4}   \bpsi_i^* \mD \mathbb{E}\left\lbrace \bn \bn^* \right\rbrace \mD \bpsi_i.
\end{align}
In case of identically distributed uncorrelated noise, i.e. $\mathbb{E} \left\lbrace \bn \bn^* \right\rbrace=\beta_n^2 \mI$, using (\ref{eq:orthogonality_on_sampling_set}), we get
\begin{align}
\label{eq::mse_gaussian_uncorrelated}
MSE_G &= \sum_{i=1}^{\abs{\F}} \frac{\beta^2_n}{\sigma_i^4}  \trace \left(  \mD \bpsi_i \bpsi_i^* \mD  \right) \nonumber \\
&= \sum_{i=1}^{\abs{\F}} \frac{\beta^2_n}{\sigma_i^4}  \trace \left( \bpsi_i^* \mD \bpsi_i \right) = \beta^2_n \sum_{i=1}^{\abs{\F}} \frac{1}{\sigma_i^2}.
\end{align}
Since the non-null singular values of the Moore-Penrose left pseudo-inverse $\left( \mB \mD \right)^+$ are the inverses of singular values of $\mB \mD$, i.e. $\lambda_i \left( \left(\mB \mD \mB \right)^+ \right) = \lambda^{-1}_i \left( \mB \mD \mB \right)$,  (\ref{eq::mse_gaussian_uncorrelated}) can be rewritten as
\begin{equation}
\label{eq::mse_frobenius}
MSE_G = \beta^2_n \, \| \left(\mB \mD \mB \right)^+ \|_F.
\end{equation}
Proceeding exactly in the same way, the  mean square error for the frame-based sampling scheme (\ref{eq::Y_def}) is
\begin{equation}
\label{eq::mse_frobenius_frame}
MSE_F = \beta^2_n \, \| \left(\mB \mY \mD \mB \right)^+ \|_F.
\end{equation}
Based on previous formulas, a possible optimal sampling strategy consists in selecting the vertices that minimize (\ref{eq::mse_frobenius})
or (\ref{eq::mse_frobenius_frame}). This aspect will be analyzed in Section \ref{Sampling strategies}.

\subsection{$\ell_1$-norm reconstruction}
Let us consider now a different observation model, where a band-limited signal $\bs \in \B$ is observed everywhere, but a subset of nodes $\S$ is strongly corrupted by noise, i.e.
\begin{equation}
\label{r=s+Dn}
\br = \bs + \mD \bn,
\end{equation} 
where the noise is arbitrary but bounded, i.e., $\| \bn \|_1 < \infty$.
This model was considered in \cite{donoho1989uncertainty} and it is relevant, for example, in sensor networks, where a subset of sensors can be damaged or highly interfered. The problem in this case is whether it is possible to recover the signal $\bs$ exactly, i.e. irrespective of noise. Even though this is not a sampling problem, the solution is still related to sampling theory. Clearly, if the signal $\bs$ is band-limited and if the indices of the noisy observations are known, the answer is simple: $\bs \in \B$ can be perfectly recovered from the noisy-free observations, i.e. by completely discarding the noisy observations, if the sampling theorem condition (\ref{|DcB|<1}) holds true. But of course, the challenging situation occurs when the location of the noisy observations is not known. In such a case, we may resort to an  $\ell_1$-norm minimization, by formulating the problem as follows
\begin{equation}
\label{eq::l1_minimization}
\tilde{\bs} = \arg \min_{\scriptsize \bs' \in \B} \| \br - \bs' \|_1.
\end{equation}
We will show next under what assumptions it is still possible to recover a band-limited signal {\it perfectly}, even without knowing exactly the position of the corrupted observations.

To start with, the following lemma, which is known as the null-space property \cite{foucart2013mathematical}, provides a necessary and sufficient condition for the convergence of (\ref{eq::l1_minimization}). 
\begin{lemma}
\label{Lemma_perfect_l1_recovery}
Given the observation model (\ref{r=s+Dn}), if for any $\bs \in \B$,
\begin{equation}
\label{eq::half_energy_inside_l1}
\| \mD \bs \|_1 < \| \mDc \bs \|_1,
\end{equation}
then the $\ell_1$-reconstruction algorithm (\ref{eq::l1_minimization}) is able to recover $\bs$ perfectly.
\end{lemma}
\begin{proof} 
To prove this, we show first that for a signal consisting of noise only, i.e. $\br=\mD \bn$, the best band-limited $\ell_1$-norm approximation $\bg$ to this signal is the zero vector. In fact,
\begin{align}
\| \mD \bn - \bg \|_1 &= \| \mD \left( \bn - \bg \right) \|_1 + \| \mDc \bg \|_1 \nonumber \\ 
&\geq \| \mD \bn \|_1 - \| \mD \bg \|_1 + \| \mDc \bg \|_1 \nonumber \\ 
&> \| \mD \bn \|_1.
\end{align}
Now, suppose instead that $\bs \neq \b0$. We can observe that
\begin{equation}
\label{signal_approx_bandlim}
\| \br - \bg \|_1 = \| \bs + \mD \bn - \bg \|_1 = \| \mD \bn + \left( \bs - \bg \right) \|_1,
\end{equation}
i.e. the best band-limited approximation $\bg$ to $\bs$ is $\bs$.
Since we proved before that, under (\ref{eq::half_energy_inside_l1}), the best band-limited approximation of $\mD \bn$ is the null vector, (\ref{signal_approx_bandlim}) is minimized by the vector $\tilde{\bs}=\bs$. 
\end{proof}
From the previous lemma it is hard to say if, for a given $\S$ and $\F$, condition (\ref{eq::half_energy_inside_l1}) holds or not. Next lemma provides such a condition.
\begin{lemma}
\label{theorem::l1_reconstruction_condition_known_S}
Given the observation model (\ref{r=s+Dn}), if
\begin{equation}
\label{eq::condition_l1}
\max_{j \in \F} \sum_i \abs{\left( \mD \mB \right)_{ij}} < \min_{j \in \F} \sum_i \abs{\left( \mDc \mB \right)_{ij}},
\end{equation}
then the $\ell_1$-reconstruction method (\ref{eq::l1_minimization}) recovers any signal $\bs \in \B$ perfectly, i.e. $\tilde{\bs} = \bs$.
\end{lemma}
\begin{proof}
Since
\begin{equation}
\sup_{\substack{\bg \in \B \\ \| \bg \|_1 = 1}} \| \mD \mB \bg \|_1 = \max_{j \in \F} \sum_i \abs{\left( \mD \mB \right)_{ij}}
\end{equation}
and
\begin{equation}
\inf_{\substack{\bg \in \B \\ \| \bg \|_1 = 1}} \| \mDc \mB \bg \|_1 = \min_{j \in \F} \sum_i \abs{\left( \mDc \mB \right)_{ij}},
\end{equation}
if (\ref{eq::condition_l1}) holds true, then 
\begin{equation}
\sup_{\bg \in \B} \frac{ \| \mD \mB \bg \|_1}{\| \bg \|_1} < \inf_{\bg \in \B} \frac{ \| \mDc \mB \bg \|_1}{\| \bg \|_1}.
\end{equation}
As a consequence, for every $\bs \in \B$, (\ref{eq::condition_l1}) implies  (\ref{eq::half_energy_inside_l1}) and then, by Lemma \ref{Lemma_perfect_l1_recovery}, it guarantees perfect recovery.
\end{proof}
Besides establishing perfect recovery conditions, Lemma \ref{theorem::l1_reconstruction_condition_known_S} provides hints on how to select the vertices to be discarded still enabling perfect reconstruction of a band-limited signal through the $\ell_1$-norm reconstruction.

An example of $\ell_1$ reconstruction based on (\ref{eq::l1_minimization}) is useful to grasp some interesting features. We consider a graph composed of 100 nodes connected by a scale-free topology \cite{albert2002statistical}. The signal is assumed to be band-limited, with a spectral content limited to the first $\abs{\F}$ eigenvectors of the Laplacian matrix. In Fig. \ref{Logan}, we report the behavior of the MSE associated to the $\ell_1$-norm estimate in  (\ref{eq::l1_minimization}), versus the number of noisy samples, considering different values of bandwidth $\abs{\F}$. As we can notice from Fig. \ref{Logan}, for any value of $\abs{\F}$, there exists a threshold value such that, if the number of noisy samples is lower than the threshold, the reconstruction of the signal is error free. As expected, a smaller signal bandwidth allows perfect reconstruction with a larger number of noisy samples.

\begin{figure}[t]
\centering
\vspace{-1cm}
\includegraphics[width=\linewidth]{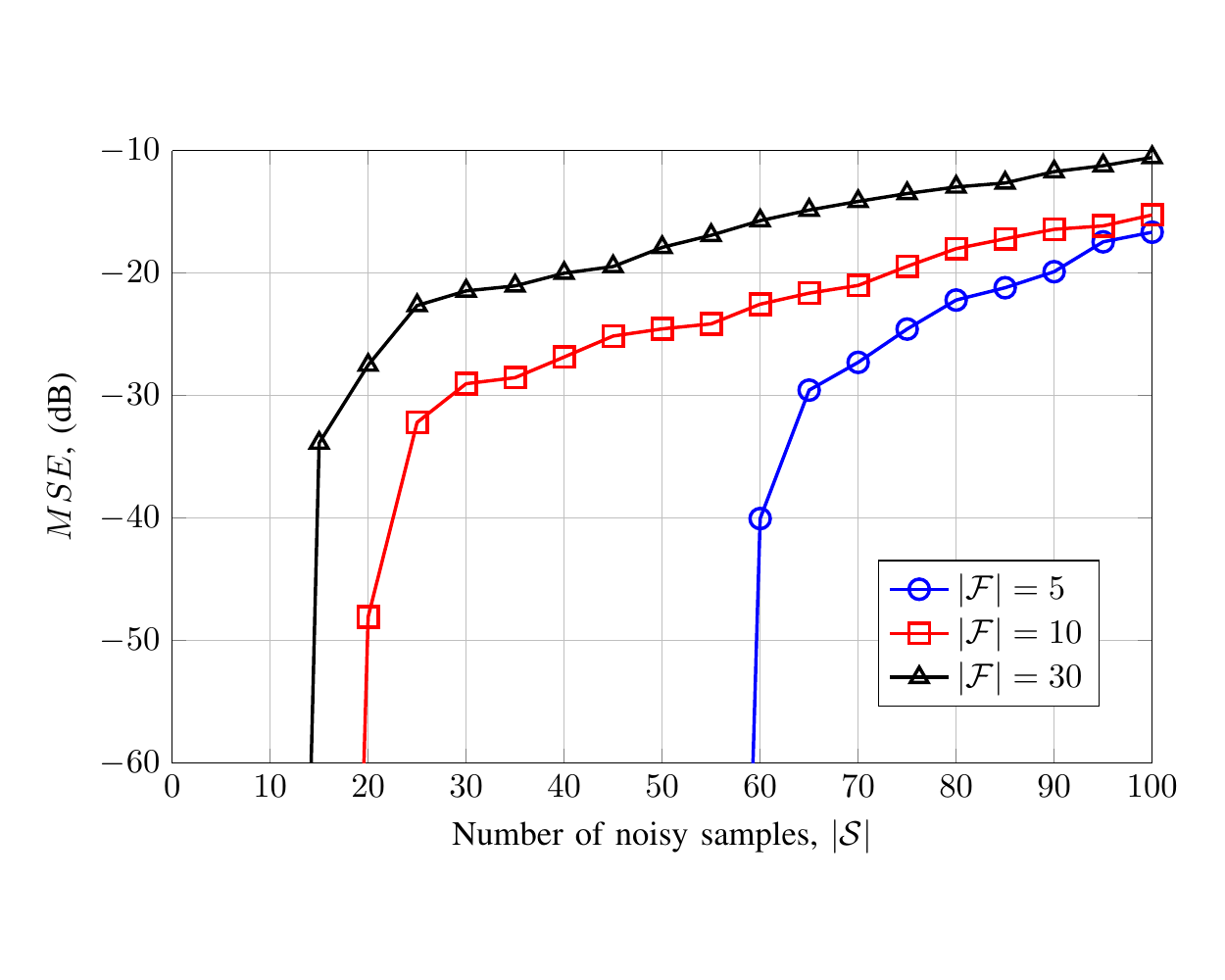}
\vspace{-1.5cm}
\caption{\small Behavior of Mean Squared Error versus number of noisy samples, for different signal bandwidths.}\label{Logan}
\vspace*{-0.5cm}
\end{figure}

We provide next some theoretical bounds on the cardinality of $\S$ and $\F$ enabling $\ell_1$-norm recovery. To this purpose, we start proving the following lemma.
\begin{lemma}
\label{theorem::sup_bound}
It holds true that
\begin{equation}
\label{sup_bound}
\sup_{ \bff \in \B} \frac{ \| \mD \bff \|_1 }{ \| \bff \|_1 } \leq \mu^2  \abs{\S} \abs{\F},
\end{equation}
where $\mu$ is defined as
\begin{equation}
\mu := \max_{ \substack{j \in \F \\ i \in \V}} \abs{u_j (i)}.
\end{equation}
\end{lemma}
\begin{proof}
Let us consider the expansion formula for $\bff \in \B$
\begin{align}
f(k) = \sum_{j \in \F} u_j(k) \sum_{i \in \V} f(i) u^*_j(i) = 
\sum_{i \in \V} f(i) \sum_{j \in \F} u_j(k) u^*_j(i)
\end{align}
which yields
\begin{equation}
\abs{f(k)} \leq \|\bff\|_{\infty} \leq \sum_{i \in \V} \abs{f(i)} \sum_{j \in \F} \mu^2 = \mu^2 \abs{\F} \| \bff \|_1,
\end{equation}
or
\begin{equation}
\label{l1_norm_bound}
\| \bff \|_1 \geq \frac{\|\bff\|_{\infty}}{\mu^2 \abs{\F}}.
\end{equation}
By combining 
\begin{equation}
\| \mD \bff \|_1 \leq \| \bff \|_{\infty} \abs{\S}
\end{equation}
with (\ref{l1_norm_bound}), we come to
\begin{equation}
\frac{ \| \mD \bff \|_1 }{ \| \bff \|_1} \leq \mu^2 \abs{\S} \abs{\F}.
\end{equation}
\end{proof}

\begin{theorem}[$\ell_1$-uncertainty]
\label{theorem::uncertainty principle l1}
Let $\bff$, $\| \bff \|_1 = 1$, be a signal $\alpha_1$-concentrated to the set of vertices $\S$, i.e. $\| \mD \bff \|_1 \geq \alpha_1$, and $\beta_1$-band-limited to the set of frequencies $\F$, i.e. $\| \mB \bff \|_1 \geq \beta_1$, then
\begin{equation}
\label{eq::uncert_princ_l1}
\left| \S \right| \left| \F \right| \geq \frac{ \left( \alpha_1 + \beta_1 - 1 \right) } { \mu^2	  \left(2-\beta_1 \right) }.
\end{equation}
\end{theorem}
\begin{proof}
If $\| \mB \bff \|_1 \geq \beta_1$, then by definition there exists a $\bg \in \B$ such that $\| \bg - \bff \|_1 \leq 1 - \beta_1$ and, for this $\bg$, we can write
\begin{equation}
\| \mD \bg \|_1 \geq \| \mD \bff \|_1 - \| \mD \left( \bg - \bff \right) \|_1\ge \| \mD \bff \|_1-1+\beta_1
\end{equation}
and
\begin{equation}
\| \bg \|_1 \leq \| \bff \|_1 + 1 - \beta_1. 
\end{equation}
Therefore
\begin{equation}
\frac{ \| \mD \bg \|_1 }{ \| \bg \|_1} \geq \frac{\| \mD \bff \|_1 - 1 + \beta_1}{ \| \bff \|_1 + 1 - \beta_1} \geq \frac{\alpha_1 + \beta_1 - 1}{2 - \beta_1}.
\end{equation}
Combining this result with the results of Lemma \ref{theorem::sup_bound}, we finally get (\ref{eq::uncert_princ_l1}). 
\end{proof}
It is worth noting that an $\ell_2$-uncertainty principle analogous to Theorem  \ref{theorem::uncertainty principle l1} may also be easily derived. Finally, we provide the condition for perfect reconstruction using (\ref{eq::l1_minimization})  when $\S$ is not known.
\begin{theorem}
\label{theorem::l1_reconstruction_unknownS}
Defining
\begin{equation}
\mu := \max_{\substack{j \in \F \\ i \in \V}} \abs{u_j (i)},
\end{equation}
if, for some {\it unknown} $\S$, we have
\begin{equation}
\label{eq::unknown_S_condition}
\abs{\S}  < \frac{1}{2 \mu^2 \abs{\F}}, 
\end{equation}
then the $\ell_1$-norm reconstruction method (\ref{eq::l1_minimization}) recovers $\bs \in \B$ perfectly, i.e. $\tilde{\bs} = \bs$, for any arbitrary noise $\bn$ present on at most $\abs{\S}$ vertices.
\end{theorem}
\begin{proof}
For a band-limited signal $\bs \in \B$ satisfying  (\ref{eq::half_energy_inside_l1}), we can also write
\begin{equation}
\label{Dg/g<1/2}
\frac{\| \mD \bs \|_1}{\| \bs \|_1} < \frac{1}{2}.
\end{equation}
On the other hand, from Lemma \ref{theorem::sup_bound} we know that the supremum of the previous ratio among all $\bs \in \B$ is upper bounded by $\mu^2 \abs{\S} \abs{\F} $. Hence, by Lemma \ref{Lemma_perfect_l1_recovery},  all band-limited signals satisfying (\ref{eq::unknown_S_condition}) satisfy also condition (\ref{Dg/g<1/2}) or, equivalently  (\ref{eq::half_energy_inside_l1}), for perfect $\ell_1$-norm recovery. 
\end{proof}

\section{Sampling strategies}
\label{Sampling strategies}
When sampling graph signals, besides choosing the right number of samples, whenever possible it is also fundamental to have a strategy indicating {\it where} to sample, as the samples' location plays a key role in the performance of reconstruction algorithms. Building on the
analysis of signal reconstruction algorithms in the presence of noise carried out in Section \ref{sec::Reconstruction from noisy observations}, a possible strategy is to select the samples' location in order to minimize the MSE. From (\ref{eq::mse_frobenius}), taking into account that
\begin{equation}
\lambda_i \left( \mB \mD \mB \right) = \sigma^2_i \left( \mB \mD \right) = \sigma^2_i \left( \mSigma \mU^* \mD \right),
\end{equation}
the problem is equivalent to selecting the right columns of the matrix $ \mSigma \mU^*$ in order to minimize the Frobenius norm of the pseudo-inverse $\left( \mSigma \mU^* \mD \right)^+$. This problem is combinatorial  and NP-hard. The problem of selecting the columns from a matrix so as to minimize the Frobenius norm of its pseudo-inverse was specifically studied for example in \cite{avron2013faster}, so that we can take advantage of those methods for our purposes. In the sequel, we provide a few alternative strategies for selecting the samples' locations.\\

\label{Numerical results}
\subsubsection{Greedy Selection - Minimization of Frobenius norm of $\left( \mSigma \mU^* \mD \right)^+$}

This strategy aims at minimizing the MSE in (\ref{eq::mse_gaussian_uncorrelated}). The method selects the columns of the matrix $\mSigma \mU^*$ so that the Frobenius norm of the pseudo-inverse of the resulting matrix is minimized. In case of uncorrelated noise, this is equivalent to minimizing $\sum_{i=1}^{\abs{\F}} 1/\sigma_i^2$. We propose a greedy approach to tackle this selection problem. The resulting sampling strategy is summarized in Algorithm 1. Note that $\S$ is the sampling set, indicating which columns to select, $\tilde{\mU}$ denotes the matrix composed by the rows of $\mU^*$ corresponding to $\F$, and the symbol $\tilde{\mU}_{\A}$ denotes the matrix formed with the columns of $\tilde{\mU}$ belonging to set $\A$.
\begin{algorithm}[h]
$\textit{Input Data}:$ $\tilde{\mU}$, rows of $\mU^*$ corresponding to $\F$;

\hspace{1.1cm} $\qquad M$, the number of samples.

$\textit{Output Data}:$ $\S$, the sampling set. \smallskip

$\textit{Function}:$ \hspace{.23cm} initialize $\S\equiv \emptyset$

\hspace{2 cm} while $|\S|<M$, set $K=\min(|\S|, |\F|)$

\hspace{2.3cm} $\displaystyle s=\arg \min_j \;\;\sum_{i=1}^{K} \frac{1}{\sigma_i^2(\tilde{\mU}_{\S\cup\{j\}})}$;

\hspace{2.3cm} $\S \leftarrow \S \cup \{s\}$;

\hspace{2cm} end

\protect\caption{\label{alg:Greedy1}\textbf{: Greedy selection based on minimum Frobenius norm of $\left( \mSigma \mU^* \mD \right)^+$}}
\end{algorithm}

\subsubsection{Maximization of the Frobenius norm of $\mSigma \mU^* \mD$}

The second strategy aims at selecting the columns of the matrix $\tilde{\mU}$ in order to maximize its Frobenius norm. Even if this strategy is not directly related to the optimization of the MSE in (\ref{eq::mse_gaussian_uncorrelated}), it leads to a very easy implementation that shows good performance in practice, as we will see in the sequel. In particular, since we have
\begin{equation}
\max_\S \;\|\tilde{\mU}\mD\|_F^2\;=\; \max_\S \;\sum_{i \in \S} \|(\tilde{\mU})_i\|^2_2,
\end{equation}
the optimal selection strategy simply consists in selecting the $M$ columns from $\tilde{\mU}$ with largest $\ell_2$-norm. \smallskip

\subsubsection{Greedy Selection - Maximization of the volume of the parallelepiped formed with the columns of $\tilde{\mU}$}

In this case, the strategy aims at selecting the set $\mathcal{S}$ of columns of the matrix $\tilde{\mU}$ that maximize the (squared) volume of the parallelepiped built with the selected columns of $\tilde{\mU}$ in $\mathcal{S}$. This volume can be computed as the determinant of the matrix $\tilde{\mU}_\mathcal{S}^* \tilde{\mU}_\mathcal{S}$, i.e. $|\tilde{\mU}_\mathcal{S}^* \tilde{\mU}_\mathcal{S}|=\prod_{i=1}^{\abs{\S}}\lambda_i(\tilde{\mU}_\mathcal{S}^* \tilde{\mU}_\mathcal{S})$, where $\lambda_i(\tilde{\mU}_\mathcal{S}^* \tilde{\mU}_\mathcal{S})$ denote the eigenvalues of $\tilde{\mU}_\mathcal{S}^* \tilde{\mU}_\mathcal{S}$, as far as $\abs{\S}\le\abs{\F}$. If $\abs{\S}$ exceeds $\abs{\F}$, we take the product of the largest $\abs{\F}$ eigenvalues. The rationale underlying this approach is not only to choose the columns with largest norm, but also the vectors more orthogonal to each other.
Also in this case, we propose a greedy approach, as described in Algorithm 2. The algorithm is similar, in principle, to the so called DETMAX algorithm mentioned in \cite{Steinberg-experimental_design1984}, but is much simpler to implement because DETMAX, at each iteration, adds and deletes points until a convergence criterion is satisfied. Our algorithm, instead, starts including the column with the largest norm in $\tilde{\mU}$, and then it adds, iteratively, the column that gives the new highest value of $|\tilde{\mU}_\mathcal{S}^* \tilde{\mU}_\mathcal{S}|$. The number of steps is then fixed and equal to the number of samples. Nevertheless, it looks suitable for graph signals because it exhibits very good performance, as shown later on.

\begin{algorithm}[h]
$\textit{Input Data}:$ $\tilde{\mU}$, rows of $\mU^*$ corresponding to $\F$;

\hspace{1.1cm} $\qquad M$, the number of samples.

$\textit{Output Data}:$ $\S$, the sampling set. \smallskip

$\textit{Function}:$ \hspace{.23cm} initialize $\S\equiv \emptyset$

\hspace{2 cm} while $|\S|<M$, set $K=\min(\abs{\S}, \abs{\F})$

\hspace{2.3cm} $\displaystyle s=\arg \max_j \;\; \prod_{i=1}^{K}\lambda_i(\tilde{\mU}_\mathcal{S}^* \tilde{\mU}_\mathcal{S})$;

\hspace{2.3cm} $\S \leftarrow \S \cup \{s\}$;

\hspace{2cm} end

\protect\caption{\label{alg:Greedy2}\textbf{: Greedy selection based on maximum parallelepiped volume}}
\end{algorithm}

\begin{figure*}[htp]
\centering
\vspace{0cm}
\centering
\begin{subfigure}[t]{0.49\linewidth}
\centering
\includegraphics[width=\columnwidth,keepaspectratio]{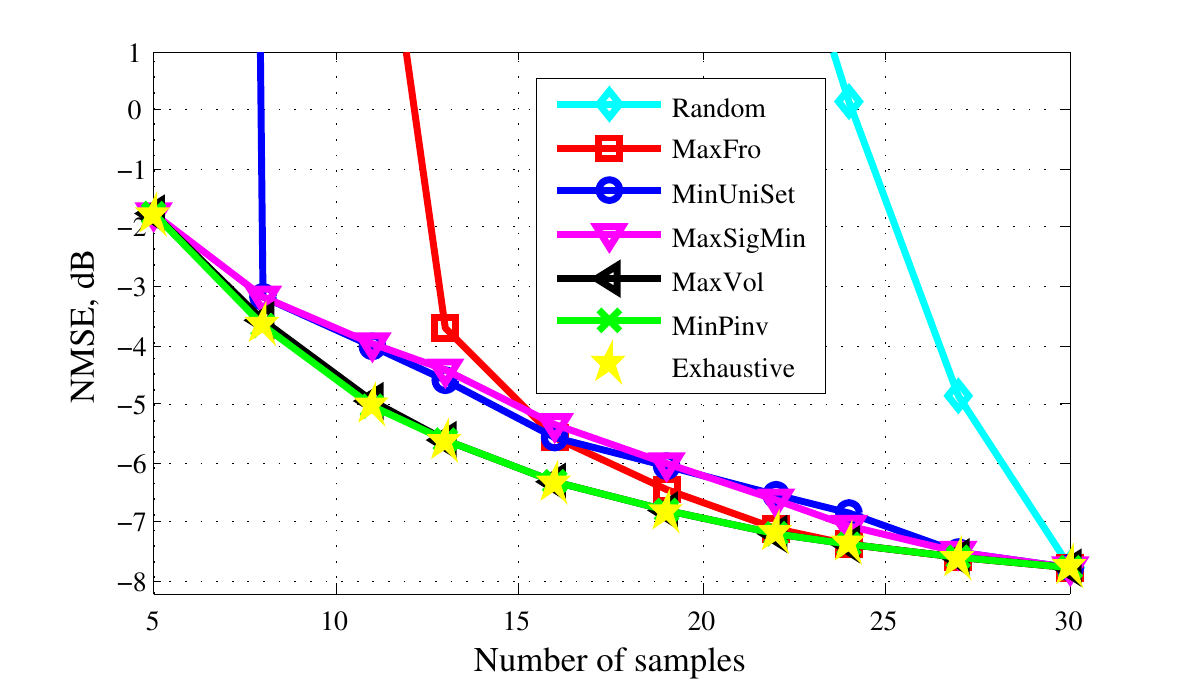}
\vspace{-0.6cm}
\caption{}
\end{subfigure}
\begin{subfigure}[t]{0.49\linewidth}
\includegraphics[width=\columnwidth,keepaspectratio]{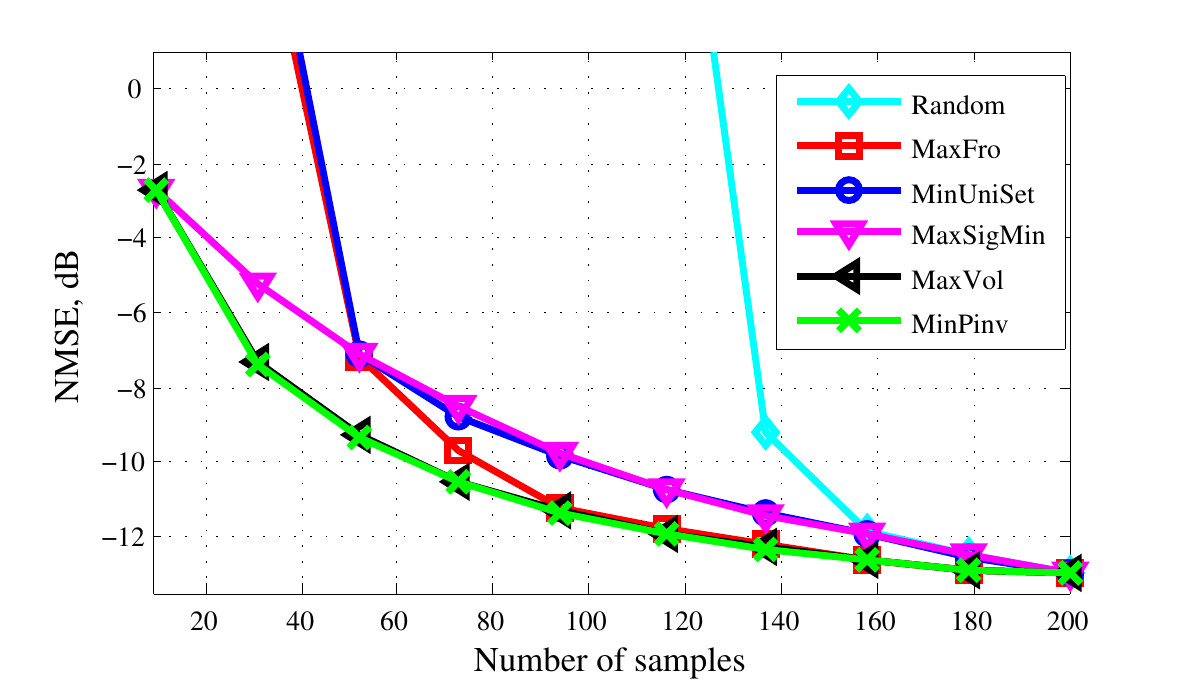}
\vspace{-0.6cm}
\caption{}
\end{subfigure}
\vspace{-0.1cm}
\caption{\small Normalized Mean Squared Error vs. number of samples for different sampling strategies and scale-free topology: \\ (a)  $N=30$, $\abs{\F}=5$; \ \ (b) $N=200$, $\abs{\F} = 10$.}
\label{fig::Sampling_SF}
\vspace{-0.4cm}
\end{figure*}

\subsubsection*{Comparison of sampling strategies}
We compare now the performance obtained with the proposed sampling strategies, with random sampling  and with two strategies proposed in the literature: 1) the method proposed in  \cite{chen2015discrete}, aimed at maximizing the minimum singular value of $\mSigma \mU^* \mD$; and 2) the approach proposed in \cite{anis2014towards}, searching for the smallest sampling set enabling the unique recovery of a band-limited signal. We test the results using a scale-free (SF) random graph model\footnote{We also tested all methods on random geometric graphs and the results were qualitatively similar.}, as this model encompasses many real world networks, see, e.g., \cite{albert2002statistical}.  Fig. \ref{fig::Sampling_SF} reports the normalized MSE (NMSE), defined as the mean square error per node, divided by the noise variance, under two configurations: (a)  $N=30$, $\abs{\F}=5$; and (b) $N=200$, $\abs{\F} = 10$.
In the case $N=30$, we report also the benchmark obtained with the exhaustive search, whereas for $N=200$ this choice is computationally too expensive. The additive noise in (\ref{r=D(s+n)}) is assumed to be an uncorrelated, zero mean Gaussian random vector with unit variance.  The results shown in the figures have been obtained by averaging over $100$ independent realizations of graph topologies. We compare six different sampling strategies, namely: (i) the random strategy, which picks nodes randomly; (ii) the greedy selection method of Algorithm 1, minimizing the Frobenius norm of $\left( \mSigma \mU^* \mD \right)^+$ (MinPinv); (iii) the Max Frobenius norm (MaxFro) strategy; (iv) the greedy selection method of Algorithm 2, maximizing the volume of the parallelepiped formed with the columns of $\mSigma \mU^* \mD$  (MaxVol); (v) the greedy algorithm  (MaxSigMin) maximizing the minimum singular value of $\mSigma \mU^* \mD$, recently proposed in \cite{chen2015discrete}; and (vi) the greedy algorithm searching for the smallest sampling set enabling unique recovery (MinUniSet), proposed in \cite{anis2014towards}. It is worth to point out that  the applicability of MinUniSet is limited to the case where the graph signal is lowpass, i.e. its GFT has a support limited on the lowest indices. Hence, for the sake of making the comparison possible, we considered a lowpass signal. However, MaxVol, MinPinv and MaxSigMin are applicable to signals whose frequency support $\F$ is {\it any} subset of $\V^*$. Furthermore, in the implementation of MinUniSet it is necessary to specify an external parameter, namely the order $k$ of the cut-off frequency (please, see  \cite{anis2014towards} for details), which affects the performance of the method. In our test, we chose a value $k=10$, as this value seemed to provide the best performance in the average. 

From Fig. \ref{fig::Sampling_SF} we observe that, as expected, the normalized mean squared error decreases as the number of samples increases. As a general remark, we can notice how random sampling performs quite poorly. This shows that, when sampling a graph signal, what matters is not only the number of samples, but also (and most important) {\it where} the samples are taken. 
Furthermore, we can notice how the proposed MaxVol and MinPinv strategies outperform all other strategies and approach very closely the optimal benchmark. The recently proposed MaxSigMin approach performs very close to the proposed MaxVol and MinPinv strategies when the number of samples is equal to its minimum value, i.e. $\abs{\S}=\abs{\F}$, but MaxVol and MinPinv outperform MaxSigMin, when the number of samples assume intermediate values between $\abs{\F}$ and $N$. Furthermore, in such a case, comparing Figs. \ref{fig::Sampling_SF}  (a) and (b), we can see how the gain increases as the number of nodes increases.  

As an example of sampling set, in Fig. \ref{IEEE118} we report an application to a real network: the IEEE 118 Bus Test Case, representing a portion of the American Electric Power System (in the Midwestern US) as of December 1962. This test graph is composed of $118$ nodes. As illustrated in \cite{pasqualetti2014controllability}, the dynamics of the power generators give rise to smooth graph signals, so that the band-limited assumption is justified, although in approximate sense. In our example, we consider a lowpass signal with  $\abs{\F} = 6$ and we take a number of samples equal to $6$. In Fig. \ref{IEEE118} we report the network structure, where the color of each node encodes the entries of the eigenvector of $\mL$ associated to the second smallest eigenvalue (these entries highlight clusters in the network, as shown in \cite{von2007tutorial}). The green squares correspond to the samples selected using either MaxVol or MinPinv strategy, which provide the same result in this case. It is interesting to notice how each method assigns two samples per cluster and puts the samples, within each cluster, quite far apart from each other. This is just an example, but it suggests an interesting conceptual link with graph independent sets, which is worth of further investigations.

Finally, we illustrate how to improve robustness to noise by using the frame-based reconstruction method. In (\ref{eq::limited_propagation}), we provided a possible choice of frame operator to be used for sampling. In the following, we show how the mean square error $MSE_F$ in (\ref{eq::mse_frobenius_frame}) behaves for different choices of graph covering sets $\N(v)$ used in (\ref{eq::limited_propagation}). 
For this example, we consider a (thorus) random geometric graph having $100$ nodes with connectivity radius $r_0=0.1883$. We consider two sampling strategies, namely: (i) the random strategy; (ii) the MaxVol strategy illustrated in Algorithm 2. Around each sample, taken at vertex $v$, the local set $\N(v)$ is composed of the nodes falling inside a ball of radius $r_1$ centered on $v$. The local sets associated to each sample can intersect each other and their union does not necessarily cover the whole graph. In Fig. \ref{fig::frames_radii2}, we show the normalized MSE as a function of $r_1$ normalized to $r_0$. We can see from Fig. \ref{fig::frames_radii2} that there exists an optimal size of covering local-sets which   minimizes the mean square error. An intuitive explanation of the behavior shown in Fig. \ref{fig::frames_radii2} is that, for small values of $r_1$, as $r_1$ increases, the local sets around each sample help reducing the MSE. However, as $r_1$ exceeds a certain threshold, the covering sets significantly overlap with each other, giving rise to a frame with more dependent vectors, in which case the MSE starts increasing again.
Furthermore, we can see how, increasing the number of samples, for a given bandwidth, the normalized MSE decreases. Finally, we can notice how the MaxVol strategy outperforms the random sampling, especially for low number of samples.
\begin{figure}[t]
\centering
\vspace{-0.2cm}
\includegraphics[width=0.8\columnwidth]{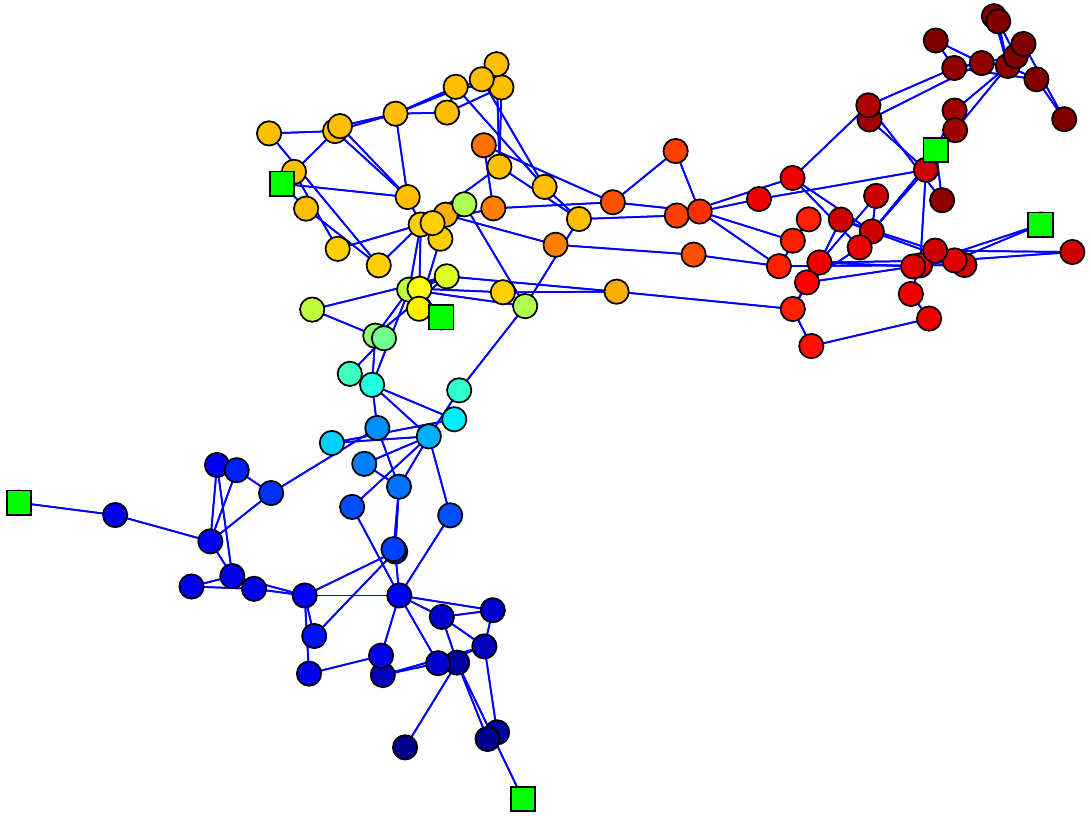}
\vspace{-0cm}
\caption{\small IEEE 118 Bus Test Case: Example of selected sampling set.}\label{IEEE118}
\vspace{-0.3cm}
\end{figure}

\section{Conclusion}
In this paper we have presented a framework for the analysis of graph signals that, starting from the localization properties over the graph and its dual domain, yields an uncertainty principle and establishes a useful conceptual link between uncertainty principle and sampling. The approach is applicable to any unitary transformation from a discrete domain to the transformed one. Besides its conceptual interest, the relation between uncertainty principle and sampling theory provides suggestions on how to identify sampling strategies and recovery algorithms robust against additive observation noise.  Interesting further developments include the extension to hypergraphs, the robustness analysis in the case of non perfectly band-limited signals and the identification of further robust recovery algorithms, including the design of optimal frame bases.

\appendices
\section{Proof of Theorem \ref{theorem::Uncertainty principle}}
Before proceeding to the proof, we introduce some useful notation and provide several results that will be used for proving Theorem \ref{theorem::Uncertainty principle}. The proof basically follows the same procedure of \cite{landau1961prolate}, where it was initially stated for continuous-time signals. 

Using the usual definition of the scalar product $\langle \ba, \bb \rangle = \ba^*  \bb $, we can define the angle between two vectors $\theta (\ba, \bb)$ as
\begin{equation}
\label{eq::angle_definition}
\theta (\ba, \bb) = \cos^{-1} \frac{\Re \langle \ba, \bb \rangle}{\| \ba \|_2 \| \bb \|_2}.
\end{equation}
By Schwartz inequality $\langle \ba , \bb \rangle \leq \| \ba \|_2 \| \bb \|_2 $ and the fact that $\abs{\Re \langle \ba , \bb \rangle} \leq \abs{ \langle \ba , \bb \rangle}$ it is clear that
\begin{equation*}
-1 \leq \frac{\Re \langle \ba , \bb \rangle}{\| \ba \|_2 \| \bb \|_2} \leq 1
\end{equation*}
and $\theta(\ba, \bb) = 0$ only if $\bb = const \cdot \ba$, i.e. when two vectors are colinear. Now, let us consider two vectors $\bff \in \B$ and $\bg \in \D$. For the beginning let us consider a fixed function $\bff \in \B$ and an arbitrary $\bg \in \D$. In this case the following lemma gives us an achievable lower bound of $\theta \left(\bff, \bg \right)$.
\begin{figure}[t!]
\vspace{-1cm}
\centering
\includegraphics[width=\columnwidth]{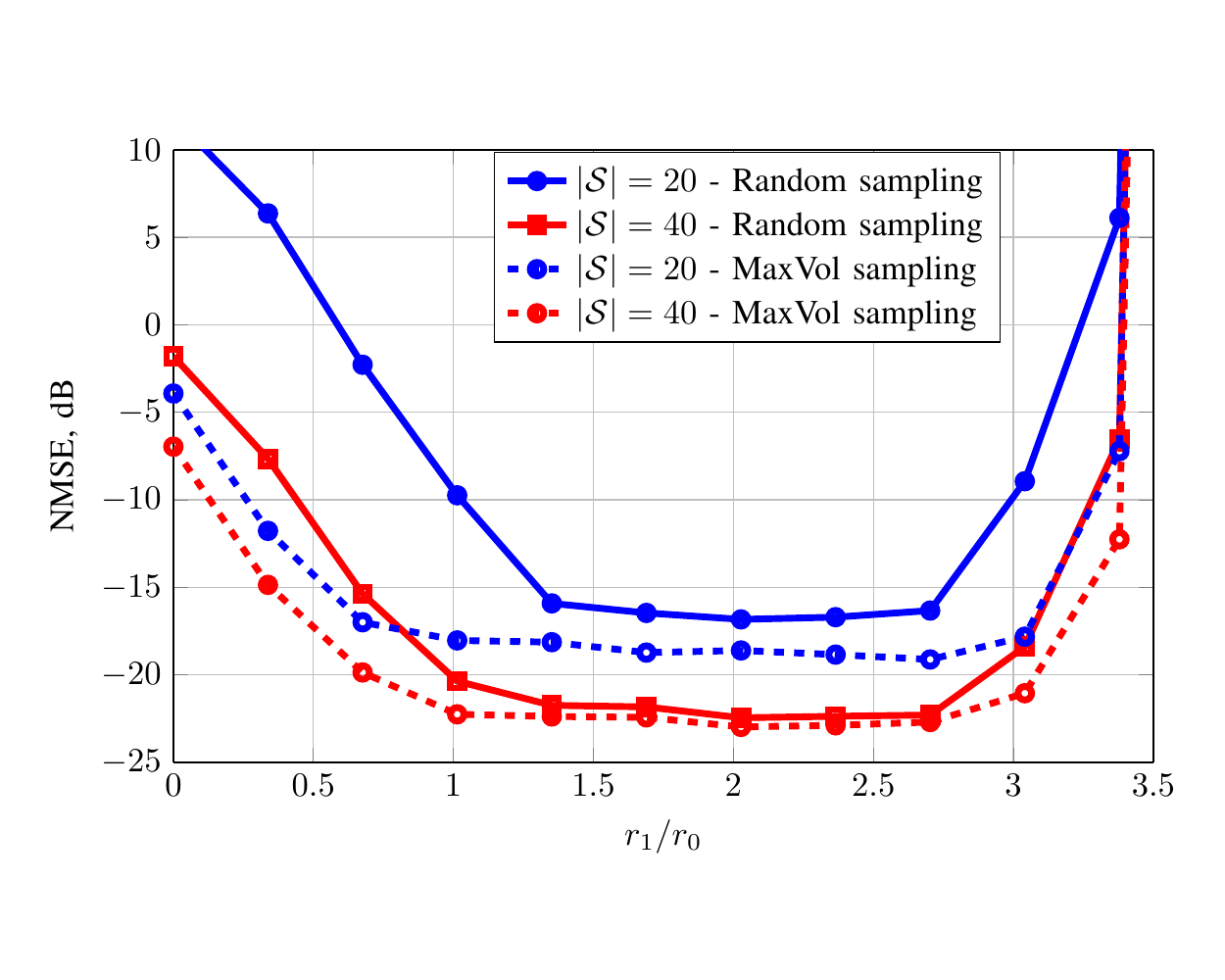}
\vspace{-1.5cm}
\caption{\small Normalized mean square error vs. the ratio $r_1/r_0$.}\label{fig::frames_radii2}
\vspace{-0.2cm}
\end{figure}
\begin{lemma}
\label{lemma::min_angle}
For a given vector $\bff$ there exists
\begin{equation}
\inf_{\bg \in \D } \theta \left( \bff, \bg \right) = \cos^{-1} \frac{\| \mD \bff \|_2}{\| \bff \|_2},
\end{equation}
which is achieved by $\bg = k \mD \bff$ for any $k > 0$.
\begin{proof}
For any $\bg \in \D$ it holds
\begin{equation*}
\Re \langle \bff, \bg \rangle \leq \abs{\langle \bff, \bg \rangle} = \abs{\langle \mD \bff, \bg \rangle}
\end{equation*}
and
\begin{equation*}
\abs{\langle \mD \bff, \bg \rangle} \leq \| \mD \bff \|_2 \cdot \| \bg \|_2.
\end{equation*}
So we can write
\begin{equation*}
\frac{ \Re \langle \bff, \bg \rangle }{\| \bff \|_2 \cdot \| \bg \|_2} \leq \frac{\| \mD \bff \|_2} { \| \bff \|_2} = \frac{ \Re \langle \bff, \mD \bff \rangle }{\| \mD \bff \|_2 \cdot \| \bff \|_2}
\end{equation*}
Taking into account that $\cos \theta$ decreases monotonically in $\left[ 0, \pi \right]$, it follows that for any $\bg \in \D$
\begin{equation*}
\theta \left( \bff, \bg \right) \geq \theta \left( \bff, \mD \bff \right),
\end{equation*}
with equality when $\bg$ and $\mD \bff$ are proportional.
\end{proof}
\end{lemma}

If the quantity 
\begin{equation}
\theta_{min} = \inf_{\substack{\bff \in \B \\ \bg \in \D}} \theta \left( \bff, \bg \right)
\end{equation}
is assumed by some specific $\bff \in \B$ and $\bg \in \D$ then we will say that $\B$ and $\D$ form the minimum angle $\theta_{min}$, which is called the first principal angle \cite{bjorck1973numerical}, and is given by the following theorem.
\begin{theorem}
\label{theorem::angle min}
The minimum angle $\theta_{min} $ between $\B$ and $\D$ exists and equals to 
\begin{equation}
\theta_{min} = \cos^{-1} \sigma_{max} \left( \mB \mD \right),
\end{equation}
and is achieved by $\bff = \bpsi_1$ and $ \bg = \mD \bpsi_1$, where $\bpsi_1$ is an eigenvector of $\mB \mD \mB$ corresponding to the eigenvalue $\sigma^2_{max} \left( \mB \mD \right)$.
\begin{proof}
Using the result of Lemma \ref{lemma::min_angle} we can write
\begin{equation*}
\inf_{ \substack{\bff \in \B \\ \bg \in \D}} \theta \left( \bff, \bg \right) = \inf_{ \bff \in \B } \cos^{-1} \frac{\| \mD \bff \|_2}{\| \bff \|_2} = \inf_{ \bff \in \B } \cos^{-1} \frac{\abs{ \langle \bff , \mD \bff \rangle}}{\| \bff \|_2},
\end{equation*}
where infimum on the left side is achieved if the infimum on the right side is achieved. Since $\cos \theta$ decreases monotonically in $\left[ 0, \pi \right]$, we can apply the result of Theorem \ref{theorem::max_concentrated_vectors}, from which it follows that infumum is achieved by the eigenvector $\bpsi_1$ of $\mB \mD \mB$ corresponding to the maximum eigenvalue $\sigma^2_{max} \left( \mB \mD \right)$. Therefore we conclude that
\begin{equation*}
\theta_{min} = \inf_{ \bff \in \B } \cos^{-1} \frac{\abs{ \langle \bff , \mD \bff \rangle}}{\| \bff \|_2} = \cos^{-1} \sigma_{max} \left( \mB \mD \right).
\end{equation*}
\end{proof}
\end{theorem}
Notice that, under perfect localization conditions, i.e. Theorem \ref{theorem_unit_eigenvalue}, $\sigma_{max} \left( \mB \mD \right) = 1$ and the minimum angle is $0$, thus implying that there are some vectors which lie in both subspaces $\B$ and $\D$. Next, we derive, without loss of generality, which values of $\beta$ are attainable for every choice of $\alpha$, assuming unit norm vectors $\bff$.

The case $\alpha = 1$ means that all the energy of signal is supported only on $\S$. According to (\ref{slep_func:max_problem}) and Lemma \ref{theorem::sing_val_bd_db} the minimally concentrated on $\F$ vector from $\D$ is the eigenvector of $\mD \mBc \mD$, corresponding to the eigenvalue $\sigma^2_{max}(\mD \mBc)$, while the maximally concentrated on $\F$ vector from $\D$ is the eigenvector of $\mD \mB \mD$, corresponding to the eigenvalue $\sigma^2_{max}(\mD\mB)$. Therefore 
\begin{equation}
\inf_{ \substack{ \bff \in \D \\ \|\bff\|_2 =1}} \beta^2 = 1-\sigma^2_{max}\left( \mD \mBc \right) 
\end{equation}
and 
\begin{equation}
\sup_{ \substack{ \bff \in \D \\ \|\bff\|_2 =1}} \beta^2 = \sigma^2_{max} \left( \mD \mB \right)
\end{equation}
for the case $\alpha = 1$. All the values in between are attainable by the function $\bff = \sum_{i=1}^K a_i \bpsi_i$ with $\sum _{i=1}^K a^2_i = 1$, where $\{\bpsi_i\}_{i=1..K}$ are the eigenvectors of $\mD \mB \mD$ belonging to $\D$ and corresponding to the eigenvalues from the interval $[1-\sigma^2_{max}\left( \mD\mBc \right),  \sigma^2_{max}\left( \mD\mB \right)]$.

Next let us consider the behavior of $\beta$ for $\alpha$ belonging to  $\alpha \in \left( 0, \, 1  \right)$. First, we will show that
\begin{equation}
\label{eq::sum_acos_goal}
\cos^{-1} \alpha + \cos^{-1} \beta \geq \cos^{-1} \sigma_{max} \left( \mB \mD \right).
\end{equation}
We can decompose any vector $\bff$ as 
\begin{equation}
\label{eq::f_equal_BplusD}
\bff = \lambda \mD \bff + \gamma \mB \bff + \bg,
\end{equation}
where $\bg$ is a vector orthogonal to both $\B$ and $\D$ and again we consider a unit norm $\bff$ with $\| \mD \bff \|_2 = \alpha$. Our goal is to find the nearest vector to $\bff$ in the space spanned by $\mD \bff$ and $\mB \bff$.

First, we calculate the inner products of (\ref{eq::f_equal_BplusD}) successively with $\bff, \mD \bff, \mB \bff$ and $\bg$ and arrive to the system of equations
\begin{align}
\left\{
\begin{array}{ll}
1 &= \lambda \alpha^2 + \gamma \beta^2 + \langle \bg, \bff \rangle, \\
\alpha^2 &= \lambda \alpha^2 + \gamma \langle \mB \bff, \mD \bff \rangle, \\
\beta^2 &= \lambda \langle \mD \bff, \mB \bff \rangle + \gamma \beta^2,  \\
\langle \bff, \bg \rangle &= \langle \bg, \bg \rangle. 
\end{array}
\right.
\end{align}
After eliminating $\langle \bg , \bff \rangle, \lambda$ and $\gamma$ from the above system we arrive to
\begin{align}
\label{eq::big_equation}
\beta^2 - 2 \Re \langle \mD \bff, \mB \bff \rangle = - \alpha^2  + \left( 1 - \frac{\abs{\langle \mD \bff , \mB \bff \rangle}^2}{\alpha^2 \beta^2} \right) \nonumber \\
- \| \bg \|_2^2 \left( 1 - \frac{\abs{\langle \mD \bff , \mB \bff \rangle}^2}{\alpha^2 \beta^2} \right).
\end{align}
According to (\ref{eq::angle_definition}) we define 
\begin{equation}
\cos \theta = \Re \frac{\langle \mD \bff, \mB \bff \rangle}{\| \mD \bff \|_2 \| \mB \bff \|_2}.
\end{equation}
Because we measure the angle $\theta$ between $\mD \bff \in \D$ and $\mB \bff \in \B$, according to Theorem \ref{theorem::angle min},
\begin{equation}
\label{eq::theta_min_BD}
\theta \geq \cos^{-1} \sigma_{max} \left( \mB \mD \right).
\end{equation}
Due to the fact that
\begin{equation}
\alpha \beta \cos \theta = \Re \langle \mD \bff, \mB \bff \rangle \leq \abs{ \langle \mD \bff, \mB \bff \rangle } \leq \alpha \beta,
\end{equation}
we can write
\begin{equation}
\label{eq::cos_theta_inequality}
0 \leq 1 - \frac{\abs{\langle \mD \bff, \mB \bff \rangle}^2}{\alpha^2 \beta^2} \leq 1 - \cos^2 \theta.
\end{equation}
In (\ref{eq::big_equation}), after introduction of $\theta$, completion of the square on the left-hand side and use of (\ref{eq::cos_theta_inequality}), we finally arrive to
\begin{equation}
\label{eq::cos_sum}
\left( \beta - \alpha \cos \theta \right)^2 \leq \left( 1 - \alpha^{2} \right) \sin^2 \theta,
\end{equation}
where equality can be achieved if and only if $\bg = \b0$ and $\langle \mD \bff, \mB \bff \rangle$ is real. Next, from (\ref{eq::cos_sum}) we can write
\begin{equation}
\beta \leq \cos \left( \theta - \cos^{-1} \alpha \right),
\end{equation}
from which it follows, using bound (\ref{eq::theta_min_BD}), that
\begin{equation}
\label{eq::acos_final_expr}
\beta \leq \cos \left( \cos^{-1} \sigma_{max} \left( \mB \mD \right)  - \cos^{-1} \alpha \right),
\end{equation}
and we immediately come to (\ref{eq::sum_acos_goal}). 
Equality in (\ref{eq::acos_final_expr}) is achieved by
\begin{equation}
\label{eq::extreme_energy_function appendix}
\bff' = p \bpsi_1 + q \mD \bpsi_1,
\end{equation}
with
\begin{equation}
\label{eq::p_def}
p = \sqrt{ \frac{1-\alpha^2}{1 - \sigma^2_{max} \left( \mB \mD \right)}},
\end{equation}
\begin{equation}
\label{eq::q_def}
q = \frac{\alpha}{\sigma_{max} \left( \mB \mD \right)} - \sqrt{\frac{1-\alpha^2}{1 - \sigma^2_{max} \left( \mB \mD \right)}}
\end{equation}
and where $\bpsi_1$ is an eigenvector of $\mB \mD \mB$ corresponding to the eigenvalue $\sigma^2_{max} \left( \mB \mD \right)$. In (\ref{eq::p_def}) and (\ref{eq::q_def}) it was supposed that $\sigma^2_{max} \left( \mB \mD \right) < 1$, because in the case $\sigma^2_{max} \left( \mB \mD \right)=1$ there exists at least one vector belonging to both $\B$ and $\D$, therefore point with $\alpha = 0$ and $\beta = 1$ belongs to $\Gamma$.

To demonstrate that $\bff'$ stays on the boundary of the uncertainty region $\Gamma$, we first rewrite (\ref{eq::acos_final_expr}) as
\begin{equation}
\label{eq::uncertainty less equal without arccos}
\beta \leq \alpha\, \sigma_{max}\left( \mB \mD \right) + \sqrt{(1-\alpha^2)(1-\sigma_{max}^2 \left( \mB \mD \right))}.
\end{equation}
Vertex and frequency energy concentrations for $\bff'$ are given by
\begin{align}
\alpha_{f'} &= \|\mD \bff' \|_2 = (p+q)\sigma_{max}\left( \mB \mD \right),\\
\beta_{f'} &= \| \mB \bff' \|_2 = p + q \sigma^2_{max}\left( \mB \mD \right).
\end{align}
Substituting $\alpha_{f'}$ and $\beta_{f'}$ in (\ref{eq::uncertainty less equal without arccos}) we immediately obtain equality.

Applying the same steps between (\ref{eq::sum_acos_goal}) and (\ref{eq::acos_final_expr}) to the operators $\mB \mDc$, $\mBc \mD$ and $\mBc \mDc$, we obtain the three remaining inequalities in (\ref{eq::uncertainty_region_Gamma}). For $\beta = 1$ and $\alpha \in \left[ 1 - \sigma^2_{max} \left( \mB \mD \right), \sigma^2_{max} \left( \mB \mD \right) \right] $  the concentrations are achievable by the eigenvectors of $\mB \mD \mB$ which belong to $\B$ and their linear combinations. Continuing by analogy one can show that all the values $\alpha$ and $\beta$ belonging to the border of $\Gamma$ (see Fig. \ref{fig:uncertainty}) are achievable. All the points inside $\Gamma$ are achievable by the functions build up from different combinations of left and right singular vectors of $\mB \mD$, $\mB \mDc$, $\mBc \mD$ and $\mBc \mDc$.

\section{Maximally concentrated dictionary for different concentrations in vertex and frequency}
Let us consider the following optimization problem
\begin{equation}
\label{eq::max_sum_energy_opt_problem}
\begin{aligned}
\bff_i &&= \  & \underset{\bff_i: \ \| \bff_i\|_2 = 1}{\arg \max} \ \gamma \| \mB  \bff_i \|^2_2 + (1-\gamma) \| \mD \bff_i \|^2_2  \\
&&& \text{s.t.} \ \ \langle \bff_i, \bff_j \rangle = 0, \ \ j\neq i,
\end{aligned}
\end{equation}
with parameter $0 < \gamma < 1$ controlling the relative energy concentration in vertex and frequency domains. The solution of (\ref{eq::max_sum_energy_opt_problem}) is given by the eigenvectors of the self-adjoint operator
\begin{equation}
\label{eq::gamma B + (1-gamma)D}
\left(\gamma \mB + (1-\gamma)\mD\right) \bff_i = \omega_i \bff_i,
\end{equation}
according to the Rayleigh-Ritz theorem. Each value of $\gamma$ corresponds to one point on the curve (\ref{beta vs alpha}) in a way that the vector $\bff_1$ maximizing (\ref{eq::max_sum_energy_opt_problem}) has energy concentrations $\left( \alpha_f, \beta_f \right)$ lying on the curve (\ref{beta vs alpha}). Hence, the solution of (\ref{eq::max_sum_energy_opt_problem}) is achieved at the tangent point of the curve (\ref{beta vs alpha}) with the line 
\begin{equation}
\label{eq::tangent line}
(1-\gamma) \alpha^2 + \gamma \beta^2 = \omega_1. 
\end{equation}
Solving the above geometric problem we obtain the pair $\alpha, \beta$ given by
\begin{align}
\label{eq::alpha f1}
\alpha_{f_1} &= \sqrt{\frac{1}{2} \left( \frac{2\gamma \left( \sigma^2_{max} -1 \right) + 1}{\sqrt{(1-2\gamma)^2-4\gamma (\gamma-1)\sigma^2_{max}}} +1 \right)}, \\
\label{eq::beta f1}
\beta_{f_1} &= \alpha_{f_1} \, \sigma_{max} + \sqrt{(1-\alpha_{f_1}^2)(1-\sigma_{max}^2 )},
\end{align}
where $\sigma_{max} := \sigma_{max}\left( \mB \mD \right)$. The eigenvalue $\omega_1$ is provided by (\ref{eq::tangent line}), i.e. $\omega_1 =(1-\gamma) \alpha_{f_1}^2 + \gamma \beta_{f_1}^2$. The first vector of the solution of (\ref{eq::max_sum_energy_opt_problem}), $\bff_1$, may be expressed in terms of $\bpsi_1$ simply by putting $\alpha_{f_1}$ into (\ref{eq::extreme_energy_function appendix}), (\ref{eq::p_def}) and (\ref{eq::q_def}).

\begin{figure}
\vspace{-0.3cm}
\begin{minipage}[b]{1.0\linewidth}
  \centering
 \centerline{\includegraphics{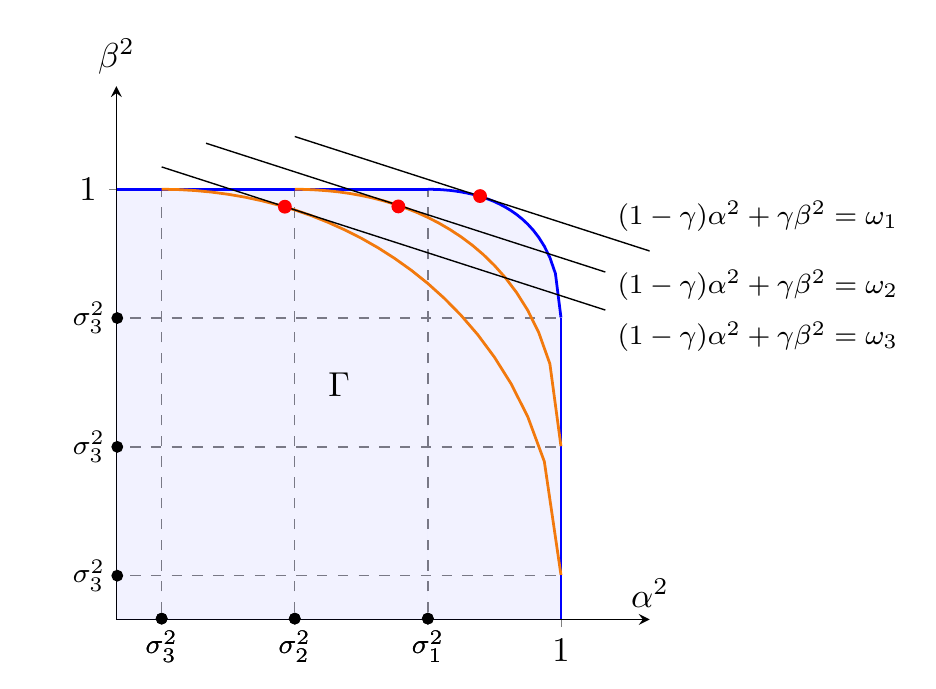}}
\end{minipage}
\vspace*{-0.6cm}
\caption{\small Position of the first three maximally concentrated vectors in the region $\Gamma$ for $\gamma = 0.75$.}
\vspace*{-0.2cm}
\label{fig:max concentrated vectors}
\end{figure}

Moreover, the first $K:=\rank \mB \mD$ orthogonal vectors $\{ \bff_i \}$ giving the solution of (\ref{eq::max_sum_energy_opt_problem}) may be constructed by substituting various $\sigma^2_{i}\left( \mB \mD \right)$ instead of $\sigma^2_{max}\left( \mB \mD \right)$ into (\ref{eq::p_def}), (\ref{eq::q_def}), (\ref{eq::alpha f1}) and then into (\ref{eq::extreme_energy_function appendix}).
To demonstrate this we consider vectors  $\bff_i$ of the form
\begin{equation}
\label{eq::extreme_energy_function_i general appendix}
\bff_i = p_i \bpsi_i + q_i \mD \bpsi_i,
\end{equation}
with
\begin{equation}
\label{eq::p_def general appendix}
p_i = \sqrt{ \frac{1-\alpha_i^2}{1 - \sigma^2_i }},
\end{equation}
\begin{equation}
\label{eq::q_def general appendix}
q_i = \frac{\alpha}{\sigma_i } - \sqrt{\frac{1-\alpha_i^2}{1 - \sigma_i^2}}
\end{equation}
and where
\begin{align}
\label{eq::alpha f1 general appendix}
\alpha_i = \sqrt{\frac{1}{2} \left( \frac{2\gamma \left( \sigma^2_i -1 \right) + 1}{\sqrt{(1-2\gamma)^2-4\gamma (\gamma-1)\sigma^2_i}} +1 \right)}.
\end{align}
For the sake of shortness we used $\sigma_i := \sigma_i \left( \mB \mD \right)$ above. First, it is easy to see that for  different $i,j=1,\ldots,K$ the vectors given by (\ref{eq::extreme_energy_function_i general appendix}) are mutually orthogonal. Secondly we want to demonstrate that vectors of the form (\ref{eq::extreme_energy_function_i general appendix}) are the eigenfunctions of $\left(\gamma \mB + (1-\gamma)\mD\right)$. We show this by direct substitution, i.e. we have
\begin{align}
(\gamma \mB + (1-\gamma)\mD) \bff_i  = & \, (\gamma p_i + \gamma \sigma_i^2 q_i) \bpsi_i \\ 
 &+ (1-\gamma)(p_i + q_i) \mD \bpsi_i. \nonumber
\end{align}
Thus $\bff_i$ is an eigenvector of $(\gamma \mB + (1-\gamma)\mD)$ if and only if the following equality holds true
\begin{equation}
\gamma + \gamma \sigma_i^2 \frac{q_i}{p_i} = (1-\gamma)\left(1 + \frac{p_i}{q_i} \right).
\end{equation}
Substituting $p_i$ and $q_i$ from (\ref{eq::p_def general appendix}) and (\ref{eq::q_def general appendix}) we easily demonstrate that equality holds. Eigenvalues $\omega_i$ are given by
\begin{equation}
\omega_i = (1-\gamma)\left(1 + \frac{p_i}{q_i} \right).
\end{equation}
In Fig. \ref{fig:max concentrated vectors} we provide an illustration showing the vertex and frequency energy concentrations of the first three $\bff_i$ for the case of $\gamma = 0.75, \sigma_1^2 = 0.85, \sigma_2^2 = 0.7$ and $\sigma_3^2 = 0.55$. The corresponding eigenvalues of (\ref{eq::gamma B + (1-gamma)D}) in this case were found to be $\omega_1 = 0.971036, \omega_2 = 0.94017$ and $\omega_3 = 0.906971$.

Using expression (\ref{eq::extreme_energy_function_i general appendix}) we have found the first $K$ vectors maximizing (\ref{eq::max_sum_energy_opt_problem}). The remaining $N-K$ vectors can be expressed in a similar way through the maximally concentrated eigenvectors of the operators $\mBc \mD \mBc$, $\mB \mDc \mB$ and $\mBc \mDc \mBc$.

\ifCLASSOPTIONcaptionsoff
  \newpage
\fi

\balance
\bibliographystyle{IEEEbib}
\bibliography{refs}

\begin{thebibliography}{10}

\bibitem{shuman2013emerging}
D.~I. Shuman, S.~K. Narang, P.~Frossard, A.~Ortega, and P.~Vandergheynst,
\newblock ``The emerging field of signal processing on graphs: Extending
  high-dimensional data analysis to networks and other irregular domains,''
\newblock {\em IEEE Signal Proc. Mag.}, vol. 30, no. 3, pp. 83--98, 2013.

\bibitem{sandryhaila2014big}
A.~Sandryhaila and J.~M.~F. Moura,
\newblock ``Big data analysis with signal processing on graphs: Representation
  and processing of massive data sets with irregular structure,''
\newblock {\em IEEE Signal Proc. Mag.}, vol. 31, no. 5, pp. 80--90, 2014.

\bibitem{pesenson2010sampling}
I.~Z. Pesenson and M.~Z. Pesenson,
\newblock ``Sampling, filtering and sparse approximations on combinatorial
  graphs,''
\newblock {\em Journal of Fourier Analysis and Applications}, vol. 16, no. 6,
  pp. 921--942, 2010.

\bibitem{sandryhaila2013discrete}
A.~Sandryhaila and J.~M.~F. Moura,
\newblock ``Discrete signal processing on graphs,''
\newblock {\em IEEE Trans. on Signal Proc.}, vol. 61, pp. 1644--1656, 2013.

\bibitem{pesenson2008sampling}
I.~Z. Pesenson,
\newblock ``Sampling in {Paley-Wiener} spaces on combinatorial graphs,''
\newblock {\em Trans. of the American Mathematical Society}, vol. 360, no. 10,
  pp. 5603--5627, 2008.

\bibitem{zhu2012approximating}
X.~Zhu and M.~Rabbat,
\newblock ``Approximating signals supported on graphs,''
\newblock in {\em IEEE Int. Conf. on Acoustics, Speech and Signal Processing
  (ICASSP)}, March 2012, pp. 3921--3924.

\bibitem{chen2015discrete}
S.~Chen, R.~Varma, A.~Sandryhaila, and J.~Kova{\v c}evi{\'c},
\newblock ``Discrete signal processing on graphs: Sampling theory,''
\newblock {\em IEEE Trans. on Signal Proc.}, vol. 63, no. 24, pp. 6510--6523,
  2015.

\bibitem{Puschel1}
M.~P{\"u}schel and J.~M.~F. Moura,
\newblock ``Algebraic signal processing theory: Foundation and 1-d time,''
\newblock {\em {IEEE Trans. Signal Process.}}, vol. 56, pp. 3572--3585, 2008.

\bibitem{Puschel2}
M.~P{\"u}schel and J.~M.~F. Moura,
\newblock ``Algebraic signal processing theory: 1-d space,''
\newblock {\em IEEE Trans. on Signal Processing}, pp. 3586--3599, 2008.

\bibitem{TsitsveroEusipco15}
M.~Tsitsvero and S.~Barbarossa,
\newblock ``On the degrees of freedom of signals on graphs,''
\newblock in {\em 2015 European Signal Proc. Conf. (Eusipco 2015)}, Sep. 2015,
  pp. 1521--1525.

\bibitem{Tsitsvero2015}
M.~Tsitsvero, S.~Barbarossa, and P.~Di~Lorenzo,
\newblock ``Uncertainty principle and sampling of signals defined on graphs,''
\newblock in {\em Proc. of Asilomar Conference on Signals, Systems, and
  Computers, Pacific Grove}, Nov. 2015.

\bibitem{chung2005laplacians}
F.~R.~K. Chung,
\newblock ``Laplacians and the {Cheeger} inequality for directed graphs,''
\newblock {\em Annals of Combinatorics}, vol. 9, no. 1, pp. 1--19, 2005.

\bibitem{agaskar2013spectral}
A.~Agaskar and Y.~M. Lu,
\newblock ``A spectral graph uncertainty principle,''
\newblock {\em IEEE Trans. on Inform. Theory}, vol. 59, no. 7, pp. 4338--4356,
  2013.

\bibitem{pasdeloup2015toward}
B.~Pasdeloup, R.~Alami, V.~Gripon, and M.~Rabbat,
\newblock ``Toward an uncertainty principle for weighted graphs,''
\newblock {\em arXiv preprint arXiv:1503.03291}, 2015.

\bibitem{benedettograph}
J.~J. Benedetto and P.~J. Koprowski,
\newblock ``Graph theoretic uncertainty principles,''
\newblock {\em
  http://www.math.umd.edu/~jjb/graph\_theoretic\_UP\_April\_14.pdf}, 2015.

\bibitem{koprowski2015finite}
P.~J. Koprowski,
\newblock {\em Finite Frames and Graph Theoretic Uncertainty Principles},
\newblock Ph.D. thesis, 2015.

\bibitem{narang2013signal}
S.K. Narang, A.~Gadde, and A.~Ortega,
\newblock ``Signal processing techniques for interpolation in graph structured
  data,''
\newblock in {\em IEEE International Conference on Acoustics, Speech and Signal
  Processing (ICASSP)}, May 2013, pp. 5445--5449.

\bibitem{anis2014towards}
A.~Anis, A.~Gadde, and A.~Ortega,
\newblock ``Towards a sampling theorem for signals on arbitrary graphs,''
\newblock in {\em 2014 IEEE Int. Conf. on Acoustics, Speech and Signal
  Processing (ICASSP)}. IEEE, 2014, pp. 3864--3868.

\bibitem{wang2014local}
X.~Wang, P.~Liu, and Y.~Gu,
\newblock ``Local-set-based graph signal reconstruction,''
\newblock {\em IEEE Trans. on Signal Processing}, vol. 63, no. 9, pp.
  2432--2444, 2015.

\bibitem{marquez2015}
A.~G. Marques, S.~Segarra, G.~Leus, and A.~Ribeiro,
\newblock ``Sampling of graph signals with successive local aggregations,''
\newblock {\em IEEE Trans. on Signal Proc. (to appear)}, 2015.

\bibitem{Chung1997}
F.~R.~K. Chung,
\newblock {\em Spectral Graph Theory},
\newblock American Mathematical Society, 1997.

\bibitem{von2007tutorial}
U.~Von~Luxburg,
\newblock ``A tutorial on spectral clustering,''
\newblock {\em Statistics and computing}, vol. 17, no. 4, pp. 395--416, 2007.

\bibitem{Folland1997}
G.~B. Folland and A.~Sitaram,
\newblock ``The uncertainty principle: A mathematical survey,''
\newblock 1997, pp. 207--238.

\bibitem{narang2013localized}
S.~K. Narang, A.~Gadde, E.~Sanou, and A.~Ortega,
\newblock ``Localized iterative methods for interpolation in graph structured
  data,''
\newblock in {\em Global Conference on Signal and Information Processing
  (GlobalSIP), 2013 IEEE}. IEEE, 2013, pp. 491--494.

\bibitem{Steinberg-experimental_design1984}
D.~M. Steinberg and W.~G. Hunter,
\newblock ``Experimental design: Review and comment,''
\newblock {\em Technometrics}, vol. 26, no. 2, pp. 71--97, May 1984.

\bibitem{avron2013faster}
H.~Avron and C.~Boutsidis,
\newblock ``Faster subset selection for matrices and applications,''
\newblock {\em SIAM Journal on Matrix Analysis and Applications}, vol. 34, no.
  4, pp. 1464--1499, 2013.

\bibitem{Ranieri-Cheibra-Vetterli-2014}
J.~Ranieri, A.~Cheibra, and M.~Vetterli,
\newblock ``Near-optimal sensor placement for linear inverse problems,''
\newblock {\em {IEEE Trans. Signal Process.}}, vol. 62, pp. 1135--1146, March
  2014.

\bibitem{KleinRandic93}
D.~Klein and M.~Randi{\'c},
\newblock ``Resistance distance,''
\newblock {\em J. Math. Chem.}, vol. 12, no. 1, pp. 81--95, 1993.

\bibitem{CoifmanMaggioni06}
R.~R. Coifman and M.Maggioni,
\newblock ``Diffusion wavelets,''
\newblock {\em Appl. Comput. Harmon. Anal.}, vol. 21, no. 1, pp. 53--94, 2006.

\bibitem{Slepian:1961:PSW}
D.~Slepian and H.~O. Pollak,
\newblock ``Prolate spheroidal wave functions, {Fourier} analysis and
  uncertainty. {I},''
\newblock {\em The Bell System Techn. Journal}, vol. 40, no. 1, pp. 43--63,
  Jan. 1961.

\bibitem{landau1961prolate}
H.~J. Landau and H.~O. Pollak,
\newblock ``Prolate spheroidal wave functions, fourier analysis and
  uncertainty, {II},''
\newblock {\em Bell System Technical Journal}, vol. 40, no. 1, pp. 65--84,
  1961.

\bibitem{albert2002statistical}
R.~Albert and A.-L. Barab{\'a}si,
\newblock ``Statistical mechanics of complex networks,''
\newblock {\em Reviews of modern physics}, vol. 74, no. 1, pp. 47, 2002.

\bibitem{Newman}
M.~Newman,
\newblock {\em Networks: An Introduction},
\newblock Oxford Univ. Press, New York, 2010.

\bibitem{duffin1952class}
R.~J. Duffin and A.~C. Schaeffer,
\newblock ``A class of nonharmonic {Fourier} series,''
\newblock {\em Trans. of the American Mathematical Society}, pp. 341--366,
  1952.

\bibitem{pesenson2015sampling}
I.~Z. Pesenson,
\newblock ``Sampling, splines and frames on compact manifolds,''
\newblock {\em GEM-International Journal on Geomathematics}, vol. 6, no. 1, pp.
  43--81, 2015.

\bibitem{donoho1989uncertainty}
D.~L. Donoho and P.~B. Stark,
\newblock ``Uncertainty principles and signal recovery,''
\newblock {\em SIAM Journal on Applied Mathematics}, vol. 49, no. 3, pp.
  906--931, 1989.

\bibitem{foucart2013mathematical}
S.~Foucart and H.~Rauhut,
\newblock {\em A mathematical introduction to compressive sensing},
\newblock Basel: Birkh{\"a}user, 2013.

\bibitem{pasqualetti2014controllability}
F.~Pasqualetti, S.~Zampieri, and F.~Bullo,
\newblock ``Controllability metrics, limitations and algorithms for complex
  networks,''
\newblock {\em IEEE Trans. on Control of Network Systems}, vol. 1, no. 1, pp.
  40--52, 2014.

\bibitem{bjorck1973numerical}
A.~Bj{\"o}rck and G.~H. Golub,
\newblock ``Numerical methods for computing angles between linear subspaces,''
\newblock {\em Mathematics of computation}, vol. 27, no. 123, pp. 579--594,
  1973.

\end{thebibliography}

\clearpage
\end{document}